\documentclass[10pt, final, twocolumn]{IEEEtran}

\usepackage{amsmath,graphicx,amssymb,pict2e,enumerate}
\usepackage{cite}
\usepackage{subfigure}
\usepackage{array}

\usepackage{bm}

\usepackage{dsfont}

\usepackage{balance}
\usepackage{amsthm}

\usepackage{comment}

\newtheorem{theorem}{\textbf{Theorem}}
\newtheorem{lemma}{\textbf{Lemma}}
\newtheorem{corollary}{\textbf{Corollary}}

\usepackage{multicol}

\usepackage{color}
\usepackage{soul}
\usepackage{booktabs}
\usepackage[bookmarks,colorlinks]{hyperref}

\newcommand{\Verify}[1]{\footnote{\textcolor{red}{Lei - #1}}}

\ifCLASSINFOpdf
\else
\fi

\begin{document}

\title{Theoretical Analysis for Extended Target Recovery in Randomized Stepped Frequency Radars}

\author{Lei Wang, Tianyao Huang*,
        and~Yimin Liu
\thanks{Part of this work \cite{Wang2019IEEEICSIDP} will be presented at the 2019 IEEE International Conference on Signal, Information and Data Processing, Chongqing, China, Dec. 2019.}
\thanks{L. Wang, T. Huang and Y. Liu are with the Department
of Electronic and Engineering, Tsinghua University, Beijing, 100084, China, e-mail:wangleixd@sina.com, \{huangtianyao, yiminliu\}@tsinghua.edu.cn.}}

\maketitle

\begin{abstract}
Randomized Stepped Frequency Radar (RSFR) is very attractive for tasks under complex electromagnetic environment. Due to the synthetic high range resolution in RSRFs, a target usually occupies a series of range cells and is called an \emph{extended target}. To reconstruct the range-Doppler information in a RSFR, previous studies based on sparse recovery mainly exploit the sparsity of the target scene but do not adequately address the extended-target characteristics, which exist in many practical applications. \emph{Block sparsity}, which combines the sparsity and the target extension, better characterizes a priori knowledge of the target scene in a wideband RSFR. This paper studies the RSFR range-Doppler reconstruction problem using block sparse recovery. Particularly, we theoretically analyze the block coherence and spectral norm of the observation matrix in RSFR and build a bound on the parameters of the radar, under which the exact recovery of the range-Doppler information is guaranteed. Both simulation and field experiment results demonstrate the superiority of the block sparse recovery over conventional sparse recovery in RSFRs.
\end{abstract}


\IEEEpeerreviewmaketitle

\section{Introduction}

\IEEEPARstart{R}{andomized} Stepped Frequency Radar (RSFR) randomly varies the carrier frequencies over wide band in a pulse-by-pulse manner. It has attracted growing attentions due to its multi-fold merits, e.g., excellent resistance to range ambiguity \cite{4137843}
, low probability of intercept and detection \cite{4338057}, and promising potential for anti-neighbour interference \cite{Huang2012RadarConf}.
In addition, while using low-cost, narrow band receiver, RSFRs coherently process with these varying carrier frequencies, synthesizing a large bandwidth and enabling high range resolution (HRR) profiling.
Since the works \cite{4137843} and \cite{Liu2000IEEERadarConf}, more and more applications in both military and civilian fields have been developed, such as RSFR-based Synthetic Aperture Radar (SAR) \cite{Yang2013TGRS} and Inverse SAR (ISAR) imaging \cite{Wang2019SensorJour}, micro-motion feature extraction \cite{Zhao2019IEEERadarConf}, cognitive radar system design \cite{Huang2014TAES} and automotive applications \cite{Hourani2017VTC}. Among these developments on RSFRs, the range-Doppler reconstruction is a common, fundamental but not simple problem.

Early works \cite{4137843,4338057,Huang2012RadarConf,Liu2000IEEERadarConf} apply the conventional matched filter for range-Doppler reconstruction, which results in high sidelobe pedestal. Weak targets could be submerged in the sidelobe of dominant ones \cite{Liu2008EL}.
As explained in \cite{Huang2018TSP}, the echoes of RSFR can be regarded as sampling of instantaneous wideband radar echoes, where each pulse occupies an instantaneous bandwidth as large as the synthetic bandwidth of RSFR. The sidelobe pedestal comes from the incomplete information in frequency domain \cite{Huang2018TSP}.
In order to alleviate the sidelobe pedestal problem, sparse recovery techniques have been introduced \cite{Liu2008EL,Huang2011Radarconf}. By exploiting the intrinsic sparsity of the target scene, sparse recovery obtains provable performance on range-Doppler reconstruction.
Particularly, \cite{Huang2018TSP} proves that, as long as the number of targets/scatterers is in the order of $\tiny{O\left(\sqrt{\frac{N}{\log{MN}}}\right)}$, where $N$ and $M$ are the numbers of transmitted pulses and carrier frequencies, respectively, exact recovery of range-Doppler parameters can be guaranteed with high probabilities. We note here that \cite{Huang2018TSP} assumes a radar target typically containing a single scatterer. This assumption holds when the range resolution is larger or comparable to the size of target.

However, when the synthetic bandwidth of RSFR becomes wider, leading to a finer HRR, the size of a target can be relatively larger than the range resolution. In this case, a target occupies a series of range cells and is called an extended target \cite{771034}.
The extended-target scene has two properties that may affect the range-Doppler recovery in RSFRs. Firstly, the number of scatterers increases sharply along with the increase of the synthetic bandwidth. As a consequence, it becomes harder to ensure the sparsity of the observing scene, which may give rise to failure of target recovery.
Secondly, extended targets exhibit additional structure. Particularly, scatterers of such target usually cluster along range, while their Doppler effects are identical.
Together with the inherent sparsity, this clustering character indicates that the extended target scene possesses \emph{block sparsity}. 
When a RSFR encounters extended targets, we apply block sparse recovery to mitigate the conceivable performance degradation of traditional sparse recovery. Utilization of block sparsity can provably yield better recovery performance than treating the signal just as being sparse in the conventional sense \cite{5319739}.

Block sparse recovery has been well studied in the literature, and successfully exploited in various applications such as face/speech recognition \cite{Wright2008TPAMI} narrow-band interference suppression \cite{liu2018block} and multiple-measurement parameter estimation \cite{Van2010TIT}. 
Many effective algorithms including greedy approaches and convex optimization methods are developed to reconstruct block-sparse signals \cite{Yuan06modelselection,Eldar2010TSP,5290295,Baraniuk2010}. Adequate researches investigate conditions under which a unique block-sparse representation of a signal can be determined by these algorithms; see \cite{Baraniuk2010, Elhamifar2012TSP} and references therein. Block sparse recovery provides good reconstruction results in practice, inspiring its utilization in RSFR applications.

In this work, we focus on theoretical analysis of block sparse recovery for RSFR. Different from previous works \cite{Eldar2010TSP,5290295,eldar2010TITAverage,Baraniuk2010,Elhamifar2012TSP,Bajwa2015} that establish generic conditions ensuring exact recovery of block-sparse signals, we prove that RSFRs are likely to satisfy these conditions under a requirement associated with radar parameters and the block sparsity of the extended-target scene.
Specifically, we begin by analyzing the specific block structures of the observation matrix in RSFR, which has not been revealed previously as the best of our knowledge, and then discuss the block coherence and spectral norm \cite{Bajwa2015} of the observation matrix. Based on the block incoherence condition \cite{Bajwa2015}, we finally prove that as long as the number of extended targets is in the scale of $\tiny{O\left(\frac{N}{M\log{MN}}\right)}$, exact reconstruction of range-Doppler parameters are guaranteed with high probabilities.

Both simulation and field experiments are carried out, and the results demonstrate that the block sparse recovery algorithms outperform the corresponding non-block sparse recovery algorithms on recovering extended targets with RSFRs.
Particularly, with measured data from practical RSFRs, block sparse recovery is shown effective to reconstruct multiple air targets and surface target in heavy clutter environment, repectively.

The rest of the paper is organized as follows. Section \ref{sec:model} formulates the signal model. Section \ref{sec:BSR} introduces some basics of traditional sparse recovery and block sparse recovery. The recovery performance analysis for RSFRs with block sparse recovery is developed in Section \ref{sec:condition}. Section \ref{sec:experiment} presents experimental results of simulated and measured (from both air and surface targets) data. Section \ref{sec:conclusion} draws a brief conclusion.

The following notation is used throughout this paper. We denote sets by upper case letters in an outline font, e.g., $\mathbb{R}$ and $\mathbb{C}$ denote the real number set and the complex number set, respectively. For $x \in \mathbb{R}$, $\left| x \right|$ and $\lfloor x\rfloor$ ($\lceil x \rceil$) represent the absolute value and the largest (smallest) integer no greater (less) than $x$, respectively. And $\delta\left(x\right)$ is the indicator function, which is 1 when $x=0$ and 0 otherwise. For $x \in \mathbb{C}$, $\left| x \right|$ represents the modulus of $x$.
We use lowercase boldface letters to denote vectors (e.g. $\bm{a}$) and uppercase boldface letters to denote matrices (e.g. $\bm{A}$). The operators $\left(\cdot\right)^*, \left(\cdot\right)^T$ and $\left(\cdot\right)^H$ represent the complex conjugate, transpose, and complex conjugate-transpose operators, respectively. For a vector $\bm{a}$, $\left[\bm{a}\right]_n$ denotes the $n$-th entry and $\|\bm{a}\|_i$ denotes the $\ell_i$ norm of $\bm{a}$, $i = 0,1,2$. For a matrix $\bm{A}$, the $\left(m, n\right)$-th element is written as $\left[\bm{A}\right]_{m,n}$ and the spectral norm of $\bm{A}$ (i.e. the maximum singular value of $\bm A$) is denoted by $\|\bm{A}\|_s$. Let $\bm{I}_N$ denote the $N$-th-order identity matrix, $\mathds{P}\left(\cdot\right)$ denote the probability of an event, and $\mathds{E}\left[\cdot\right]$ represent the expectation of a random argument.

\section{RSFR Signal Model}
\label{sec:model}
In this section, we present the signal model of RSFR, following the presentation in \cite{Huang2018TSP}. However, unlike \cite{Huang2018TSP}, which models a target as a single scatterer, we consider an extended-target model, in which each target contains multiple scatterers moving at identical velocity. Under the extended target model, we reveal that the target scene possesses block sparsity, which inspires the application of block sparse recovery algorithms, different from the use of traditional sparse recovery in previous work \cite{Huang2018TSP}. We review the transmit model of RSFR in Subsection \ref{subsec:transmission}, and detail the receive model in Subsection \ref{sec:model:signal}, which is then recast in matrix form as present in Subsection \ref{subsec:model:matrix form}.

\subsection{Transmission of RSFR}
\label{subsec:transmission}

In a RSFR, there are $N$ single-frequency sinusoidal pulses transmitted during a Coherent Processing Interval (CPI). For the $n$-th pulse, $n \in \mathbb{N}:=\{0,1,\dots,N-1\}$, the carrier frequency is randomly varied as $f_n = f_c + C_n \Delta{f}$, where $f_c$ is the initial carrier frequency, $\Delta{f}$ is the frequency step interval and $C_n \in \mathbb{M}:=\{0,1,\dots,M-1 \}$ is the randomized modulation code. Thus, the $n$-th transmitted pulse can be expressed as
\begin{equation}
\label{Eq:TransPulse}
s_{\rm T}\left( n, t \right)= \text{rect}\left(\frac{t - nT_r}{T_p} \right)e^{j2\pi f_n\left(t - nT_r\right)},
\end{equation}
where $T_r$ is the Pulse Repetition Interval (PRI), $T_p$ is the pulse width and $\text{rect}\left( t \right)$ is the rectangular function defined as
\begin{equation}
\text{rect}\left( t \right) = \left\{\begin{array}{ll}
1& 0 \leq t \leq 1,\\
0& \text{otherwise}.
\end{array}\right.
\end{equation}
Here, we assume that the modulation codes $C_n$ are independently identically distributed random variables with uniform density over $\mathbb{M}$, i.e., $C_n \sim U\left( \mathbb{M}\right)$.
In RSFR, the instantaneous bandwidth of each pulse, denoted by $B_0 := 1/T_p$, is usually narrow. The narrow bandwidth leads to a low Coarse Range Resolution (CRR), i.e.,  $\frac{c}{2B_0}$, where $c$ is the speed of light. By synthesizing echoes of different frequencies, we obtain a larger synthetic bandwidth, $B = M\Delta f$, which refines the range resolution to $\frac{c}{2B}$.

\subsection{Radar Returns Model}
\label{sec:model:signal}

We then derive the expressions of received echoes, which are delays of the transmissions. We begin by considering a single ideal scatterer with complex scattering coefficient $\gamma$. Multiple-scatterer scenario is a simple extension, and will be discussed later in this section. Assume that the scatterer is moving along the radar line of sight with a constant velocity $v$ and an initial range $R$. Let $\tau\left( t \right) := \frac{2\left(R + vt\right)}{c}$ represent the time delay at the time instant $t$. Under the ``stop-and-go'' model \cite{richards2005fundamentals}, it holds that $\tau\left( t \right) \approx \tau\left( nT_r \right)$.
Then, the received echo can be written as
\begin{equation}
\label{Eq:RecievedPulse}
{s}_{\rm R}\left( n, t \right) = \gamma s_{\rm T}\left( n, t - \tau\left( t \right) \right) \approx \gamma s_{\rm T}\left( n, t - \tau\left( nT_r \right) \right).
\end{equation}
The RF echo of each pulse, ${s}_{\rm R}\left( n, t \right)$, is down-converted to the baseband by its corresponding carrier frequency, e.g., $e^{j2\pi f_n t}$ for the $n$-th pulse. After down-conversion, the baseband echo is represented by
\begin{equation}
\label{Eq:BaseBandEcho}
\begin{split}
\overline{s}_{\rm R}\left( n, t \right) &\!\! =  {s}_{\rm R}\left( n, t \right) \cdot e^{-j2\pi f_n \left(t - nT_r\right)}\\
&\!\! = \gamma \text{rect}\left(\frac{t - nT_r - \tau\left( nT_r \right)}{T_p} \right) e^ {-j2\pi f_n \tau\left( nT_r \right)}\\
&\!\! = \overline{\gamma} \text{rect}\left(\frac{t - nT_r - \tau\left( nT_r \right)}{T_p} \right) \cdot e^ {j2\pi f_R C_n + j2\pi f_v \xi_n n}
\end{split}
\end{equation}
where $\overline{\gamma} = \gamma e^{-j2\pi f_c \tau\left( nT_r \right)}$ and $\xi_n := 1 + \frac{C_n\Delta f}{f_c}$. We regard $f_R := -\frac{2\Delta f R}{c}$ and $f_v := -\frac{2f_cvT_r}{c}$ as the range frequency and velocity frequency, respectively.

For each pulse, we then sample the baseband echoes $\overline{s}_{\rm R}\left( n, t \right)$ at time instants, $t = nT_r + l_s/f_s$, $l_s = 0, 1, \cdots, \lfloor T_r f_s\rfloor$. We use the Nyquist sampling rate, i.e., $f_s = B_0$, so that each sample corresponds to a CRR bin. Echoes from these bins are processed identically and individually. Without loss of generality, suppose that the $l$-th CRR bin contains the scatterer, and the scatterer stays inside the bin during the CPI. In the rest of paper we focus on this single $l_s$-th CRR bin.
By substituting $t = nT_r + l_s/f_s$ into \eqref{Eq:BaseBandEcho}, we obtain the echo sequence
\begin{equation}
\label{Eq:BaseBandEcho slow time}
\begin{split}
s_{\rm R}\left( n \right) = \overline{s}_{\rm R}\left( n, t \right)|_{t = nT_r + l_s/f_s} =  \overline{\gamma}e^ {j2\pi f_R C_n + j2\pi f_v \xi_n n}.
\end{split}
\end{equation}

The model \eqref{Eq:BaseBandEcho slow time} above is derived in the single target/scatterer case. We now extend it to the scenario that a CRR bin contains $K$ targets and the $k$-th target consists of $P_k$ scatterers, $k=0,1, \cdots, K-1$.
Denote by $v_k$ and $f_{v_k}$ the velocity of the $k$-th target and its velocity frequency, respectively. Let $\{{\bar \gamma}_{ki}\}_{i=0}^{P_k-1}$, $\{R_{ki}\}_{i=0}^{P_k-1}$  and $\{f_{R_{ki}}\}_{i=0}^{P_k-1}$ be the scattering coefficients, initial ranges and the corresponding range frequencies of the scatterers contained in the $k$-th target, respectively. We note that these $P_k$ scatterers have the same velocity $v_k$, since they belong to the same target.
The received signal is modeled as the superimposed echoes from all the scatterers belonging to these $K$ targets,
\begin{equation}
\label{Eq:Multi Target RecievedPulse}
s_{\rm R}\left( n \right) =  \sum_{k=0}^{K-1} \sum_{i=0}^{P_k-1}{\overline{\gamma}_{ki} e^ {j2\pi f_{R_{ki}} C_n + j2\pi f_{v_k} \xi_n n}}.
\end{equation}
Here, $\{{\bar \gamma}_{ki}\}$, $\{f_{R_{ki}}\}$  and $\{f_{v_k}\}$, representing the intensity, range and Doppler parameters, respectively, are unknown and should be estimated from the sampled echoes $s_{\rm R}(n)$, $n \in \mathbb{N}$.

In the next subsection, we will rewrite \eqref{Eq:Multi Target RecievedPulse} in a matrix form and reveal the connection between range-Doppler reconstruction and block sparse recovery.

\subsection{Matrix form model with block structure}
\label{subsec:model:matrix form}

Following \cite{Huang2018TSP}, in this subsection, we reformulate the echo model \eqref{Eq:Multi Target RecievedPulse} in matrix form and recast the range-Doppler estimation as a sparse recovery problem. Different from the previous work that exploits non-block sparsity \cite{Huang2018TSP}, we will emphasize the natural block sparsity that appears in wideband RSFR.

Stacking the echoes forms the measurement vector $\bm{y} \in \mathbb{C}^{N}$ with $n$-th entry given by $\left[\bm{y}\right]_n = s_{\rm R}\left( n \right)$.

We then discretize the continuous range frequency and  velocity frequency parameters, $f_R$ and $f_v$, into finite grid points. Note that $\left( f_R, f_v\right)$ is unambiguous in the region $[0,1)^2$ and the resolutions of $f_R$ and $f_v$ are $1/M$ and $1/N$, respectively. We discretize $f_R$ and $f_v$ at the rates of $1/M$ and $1/N$, respectively, leading to the set of grid points, $\{\frac{p}{M} \}_{p \in \mathbb{M}} \times \{\frac{q}{N} \}_{q \in \mathbb{N}}$.
Under the assumption that all the scatterers are located precisely on the grid points, we denote by $\bm \Gamma \in \mathbb{C}^{M \times N}$ the scattering intensities corresponding to the grid points. The $(p,q)$-th entry of $\bm \Gamma$, denoted by $\Gamma_{p,q}$, is given by
\begin{equation}
\label{eq:x}
  \Gamma_{p,q} := \left\{\begin{array}{cl}
\!\! \sqrt{N}\overline{\gamma}_{ki}& \!\! \text{if} \ \exists \left(k, i\right), \left(f_{R_{ki}},f_{v_{k}}\right) = \left(\frac{p}{M},\frac{q}{N}\right)\\
\!\! 0& \!\! \text{otherwise},
\end{array}\right.
\end{equation}
where $\sqrt{N}$ is a normalization factor, representing the gain of coherent processing with $N$ pulses.

We denote by $\bm x_q$ the $q$-th column of $\bm \Gamma$, i.e., $\bm{x}_q := \left[\Gamma_{0,q}, \Gamma_{1,q},\cdots, \Gamma_{M-1,q}\right]^T \in \mathbb{C}^{M}$, which corresponds to a target with velocity frequency $\frac{q}{N}$ and represents the HRR profiles of the target. Vectorization of $\bm \Gamma$ yields $\bm{x} \in \mathbb{C}^{MN}$, i.e.,
\begin{equation}
\label{Eq: block signal}
    \bm{x}:=\left[\bm x_{0}^T, \bm x_{1}^T, \cdots, \bm x_{N-1}^T \right]^T,
\end{equation}
where the HRR profile of each target can be regarded as a block of $\bm x$.

Since there are generally only a few targets in a certain CRR, the observed scene is often sparse. Particularly, due to the block structure indicated in \eqref{Eq: block signal}, only a few blocks in $\bm x$ are nonzero, which reveals that in RSFS the scene possesses \emph{block sparsity}. This additional structure inspires us to apply block sparse recovery in RSFR instead of the canonical sparse recovery. Exploiting the block sparsity leads to better range-Doppler reconstruction performance, as will be discussed later by the theoretical analysis and simulation/field experiments, presented in Section \ref{sec:condition} and \ref{sec:experiment}, respectively.

We now arrange \eqref{Eq:Multi Target RecievedPulse} in matrix form as
\begin{equation}
\label{Eq: matrix form model noiseless}
\bm{y} = \bm \Psi \bm{x},
\end{equation}
where $\bm{\Psi} \in \mathbb{C}^{N \times MN}$ is referred to as the observation matrix.
Consistent with the definition of $\bm x$, $\bm{\Psi}$ is divided into $N$ blocks, i.e.,
\begin{equation}
\label{eq:psi with blocks}
  \bm \Psi := \left[\bm \Psi_0, \bm \Psi_1, \cdots, \bm \Psi_{N-1} \right],
\end{equation}
and each block $\bm \Psi_q \in \mathbb{C}^{N \times M}$, $q \in \mathbb{N}$, corresponds to a unique velocity frequency $\frac{q}{N}$.
There are $M$ columns in a block $\bm \Psi_q$ and we denote by $\bm \psi_{p,q}$ the $p$-th column, i.e., $\bm{\Psi}_q = \left[ \bm{\psi}_{0,q}, \bm{\psi}_{1,q}, \bm{\psi}_{2,q}, \cdots, \bm{\psi}_{M-1,q} \right]$. From \eqref{Eq:BaseBandEcho slow time} and \eqref{eq:x}, the $n$-th entry of $\bm{\psi}_{p,q}$ is given by
\begin{equation}
\label{Eq: observation vector}
\left[\bm{\psi}_{p,q}\right]_n = \frac{1}{\sqrt{N}}e^ {j\frac{2\pi p}{M} C_n + j\frac{2\pi q}{N}\xi_n n},
\end{equation}
where the factor $\frac{1}{\sqrt{N}}$ normalizes the observing vectors so that $\bm{\psi}^H_{p,q}\bm{\psi}_{p,q} = 1$.

In a noisy circumstance, \eqref{Eq: matrix form model noiseless} is rewritten as:
\begin{equation}
\label{Eq:matrix form model}
\bm{y} = \bm{\Psi{x}} + \bm w,
\end{equation}
where $\bm{w}$ is the additive white Gaussian noise with a noise power $\sigma^2$, i.e., $\bm{w}\sim \mathcal{CN}\left(\bm{0},\sigma^2 \bm{I}_N\right)$.

In (\ref{Eq: matrix form model noiseless}) and (\ref{Eq:matrix form model}), $\bm y$ and $\bm \Psi$ are given, while $\bm x$ is unknown and yet should be recovered. When $\bm x$ is reconstructed by solving the linear equation (\ref{Eq: matrix form model noiseless}) or (\ref{Eq:matrix form model}), the HRR profiles and velocity parameters of targets are recovered from the indices of nonzero elements in $\bm{x}$.
Because the dimension of the observations is less than that of the unknown vector $\bm{x}$, i.e. $N<MN$, the problem is under-determined, which inspires the use of sparse recovery or compressed sensing, as discussed previously in \cite{Huang2018TSP}. In this paper, observing the additional block sparsity of the target scene, we apply the block sparse recovery algorithms with the expectation of achieving better reconstruction performance. Both traditional and block sparse recovery will be briefly reviewed in the next section.

\section{Sparse Recovery and Block Sparse Recovery}
\label{sec:BSR}
We first introduce some basic concepts of non-block sparse recovery in Subsection \ref{subsec:sparse:non-block}, and then briefly review block sparse recovery in Subsection \ref{subsec:sparse:block}.

\subsection{Sparse recovery}
\label{subsec:sparse:non-block}
Sparse recovery aims to solve the under-determined problems such as $\bm{y = \Psi{x}}$. In particular, it assumes that $\bm{x}$ is sparse, i.e., there are only a few nonzero entries in $\bm{x}$, and seeks for the sparsest representation of $\bm{y}$ by minimizing the $\ell_0$ "norm"
\begin{equation}
\label{Eq:L0 model}
    \hat{\bm{x}} = \arg \min_{\bm{x}} \| \bm{x} \|_0 , \text{s.t.} \ \bm{y} = \bm{\Psi x}.
\end{equation}
Since the $\ell_0$ optimization is generally NP-hard, many strategies  have been proposed to reduce the computational complexity including greedy approaches and more efficient $\ell_1$ minimization, i.e.,
\begin{equation}
\label{Eq:L1 model}
    \hat{\bm{x}} = \arg \min_{\bm{x}} \| \bm{x} \|_1 , \text{s.t.} \ \bm{y} = \bm{\Psi x}.
\end{equation}

There are many works addressing conditions under which \eqref{Eq:L1 model} has a unique solution; see \cite{Eldar2012} and references therein. Most of these researches rely on the mutual coherence or restricted isometry property (RIP) of the measurement matrix $\bm \Psi$, the sparsity level (the number of nonzero elements in $\bm x$) as well as the dimensions of the problem. For example, in RSFR, a specific application of sparse recovery, \cite{Huang2018TSP} proves that \eqref{Eq:L1 model} guarantees the successful recovery of $\bm x$ with high probability (with respect to the random selections of carrier frequencies) when the number of nonzero entries in $\bm x$ is in the order of $O\left(\sqrt{\frac{N}{\log MN}} \right)$.
We later introduce block sparse recovery, which can provably yield better reconstruction properties than treating $\bm x$ being sparse in this conventional sense.

\subsection{Block sparse recovery}
\label{subsec:sparse:block}

Block sparse recovery assumes that nonzero elements appear in a few blocks. And the vector $\bm{x}$ is said $K$-block sparse, if there are at most $K$ nonzero blocks. As discussed in Subsection \ref{subsec:model:matrix form}, block sparsity naturally arise in RSFR when targets are extended in range. Each block of $\bm x$, as defined in \eqref{Eq: block signal}, represents the HRR profiles of an extended target moving at a specific velocity. As noted in \cite{Elhamifar2012TSP}, in general, a block-sparse vector is not necessarily sparse and vice versa.

Block sparse recovery turns to minimize the number of nonzero blocks 
in $\bm{x}$ by solving the following optimization problem
\begin{equation}
\label{Eq:Lp0 model}
    \hat{\bm{x}} = \arg \min_{\bm{x}} \| \bm{x} \|_{i,0} , \text{s.t.} \ \bm{y} = \bm{\Psi x},
\end{equation}
where $i\geq 0$ and $\| \bm{x} \|_{i,0} := \sum_{q=0}^{N-1}{\left\|\|\bm{x}_q\|_i\right\|_0}$ is the mixed $\ell_{i,0}$ norm. However, solving \eqref{Eq:Lp0 model} is still NP-hard \cite{Elhamifar2012TSP}.
To efficiently solve \eqref{Eq:Lp0 model}, convex relaxation that applies $\ell_{2,1}$ norm can be used, i.e.,
\begin{equation}
\label{Eq:noiseless L21 model}
    \hat{\bm{x}} = \arg \min_{\bm{x}} \| \bm{x} \|_{2,1} , \text{s.t.} \ \bm{y} = \bm{\Psi x},
\end{equation}
where $\| \bm{x} \|_{2,1} := \sum_{q=0}^{N-1}{\|\bm{x}_q\|_2}$. In a noisy case \eqref{Eq:matrix form model}, a so-called Block-Lasso \cite{Yuan06modelselection} method is usually adopted as
\begin{equation}
\label{Eq:noisy L21 model}
    \hat{\bm{x}} = \arg \min_{\bm{x}} \frac{1}{2} \|\bm{y} - \bm{\Psi x}\|_2^2 + \lambda\| \bm{x} \|_{2,1},
\end{equation}
where $\lambda > 0$ is the weight coefficient for regularization.

By generalizing the notion of coherence or RIP to block setting, many works study  conditions under which \eqref{Eq:noiseless L21 model} yields correct reconstruction of block sparse $\bm x$, including \cite{Eldar2010TSP,5290295,eldar2010TITAverage,Baraniuk2010,Elhamifar2012TSP,Bajwa2015} to name a few.

Among these works, we adopt the average-case analysis framework provided in a more recent paper \cite{Bajwa2015}, for its explicitly computable conditions on $\bm \Psi$ in contrast to the classical setup.  
As opposite to the conventional analyses that consider to recover an arbitrary $K$-block-sparse $\bm x$, \cite{Bajwa2015} resorts to an average-case analysis by imposing a mild statistical prior on $\bm x$. We repeat a concise version of these mild statistical constraints as the following:
\begin{itemize}
    \item[M1)] The block support of $\bm{x}, \mathbb{T} := \left\{ q: \bm{x}_q \ne \bm 0\right\}$, has a uniform distribution over the all $\left( \substack{N \\ K}\right)$ possible $K$-subsets of $\mathbb{N}$,
    \item[M2)] Entries in $\bm x$ have zero median: $\mathds{E}\left[ {\rm sign}(\bm x)\right] = \bm 0$, where ${\rm sign}(x) = x/|x|$ denotes entry-wise sign operation, and
    \item[M3)] Nonzero blocks of block-sparse signal $\bm{x}$ have statistically independent ``directions''.
\end{itemize}

We also inherit from \cite{Bajwa2015} the definitions of intra-block coherence, i.e.,
\begin{equation}
\label{Eq:intra-block coherence}
    \mu_I := \max_{ q \in \mathbb{N} } \| \bm{\Psi}_q^{H}\bm{\Psi}_q - \bm{I}_M \|_s,
\end{equation}
and inter-block coherence, given by
\begin{equation}
\label{Eq:inter-block coherence}
    \mu_B := \max_{ q_1, q_2 \in \mathbb{N}, q_1 \neq q_2 } \| \bm{\Psi}_{q_1}^{H}\bm{\Psi}_{q_2} \|_s.
\end{equation}
With these definitions, \cite{Bajwa2015} provides the following theorem that guarantees the unique solution of block sparse recovery \eqref{Eq:noiseless L21 model}.
\begin{theorem}[\hspace{1sp}\cite{Bajwa2015}]
\label{theorem:origin}
Suppose that $\bm{x}$ is $K$-block sparse, drawn according to the statistical model M$1$-M$3$, and is observed according to (\ref{Eq: matrix form model noiseless}). Then, as long as the block coherence of the matrix $\bm{\Psi}$ satisfy
\begin{equation}
\label{Eq: original sufficient condition}
\begin{split}
17\sqrt{\frac{K\log{\left(MN\right)}\left(1 + \mu_I\right)}{N}}&\|\bm{\Psi}\|_s \\
+ 48\mu_B\log\left(MN\right) +& \frac{2K}{N}\|\bm{\Psi}\|^2_s + 3\mu_I \leq \frac{1}{4},
\end{split}
\end{equation}
solving \eqref{Eq:noiseless L21 model} results in $\hat{\bm{x}} = \bm{x}$ with probability at least $1 - 4\left( MN \right)^{-4\log2}$, with respect to the random choice of the subset $\mathbb{T}$.
\end{theorem}
\begin{proof}
See  \cite[(5) and Thm. 1 and 2]{Bajwa2015}.
\end{proof}

Next, based on Theorem \ref{theorem:origin}, especially the so-called block incoherence condition \eqref{Eq: original sufficient condition}, we analyze the block coherence of $\bm \Psi$, and establish the corresponding unique recovery condition on the block sparsity $K$ in RSFR.
The condition \eqref{Eq: original sufficient condition} imposes a joint constrain on the block coherence $\mu_I$, $\mu_B$, the spectral norm $\|\bm{\Psi}\|_s$ and the block sparsity $K$. We note that these parameters $\mu_I$, $\mu_B$ and $\|\bm{\Psi}\|_s$ are not trivial extensions of the traditional coherence used in \cite{Huang2018TSP}, but rely on the structure of the block matrices $\bm \Psi_q$. 
The novelty vis-$\grave{a}$-vis the reference \cite{Huang2018TSP} lies in revealing and leveraging the particular structure of these blocks in $\bm \Psi$.


\section{Performance Analysis for RSFR}
\label{sec:condition}
In this section, we analyze the range-Doppler reconstruction performance of RSFR using block sparse recovery based on Theorem \ref{theorem:origin}, which involves the block coherence $\mu_I$ and $\mu_B$, and the spectral norm of the overall observation matrix $\left\| \bm \Psi\right\|_s$. Since the carrier frequency for each radar pulse is randomized, the observation matrix $\bm{\Psi}$ is random. As a consequence, we start analyzing the probabilistic characters of $\mu_I$ and $\mu_B$ in Subsection \ref{subsec:uIuB}, followed by the calculation of $\left\| \bm \Psi\right\|_s$ present in Subsection \ref{subsec:psis}. Given these results on $\mu_I$, $\mu_B$ and  $\left\| \bm \Psi\right\|_s$, we then develop conditions that ensure unique recovery exploiting block sparsity in Subsection \ref{subsec:bound}.

In order to facilitate the analysis, we follow the typical setting in \cite{Costas1984IEEEProceedings},
assuming $\xi_n = 1$, throughout this section, so that the $n$-th entry of the observation vector \eqref{Eq: observation vector} can be simplified as follows
\begin{equation}
\label{Eq: simplified observation vector}
\left[\bm{\psi}_{p,q}\right]_n = \frac{1}{\sqrt{N}}e^ {j\frac{2\pi p}{M} C_n + j\frac{2\pi q}{N} n}.
\end{equation}
In fact, this assumption is to neglect the Doppler-shift differences of different carrier frequencies, which holds when the relative bandwidth $B/f_c$ is negligible. However, when we consider a RSFR with large (synthetic) bandwidth, this approximation does not usually hold unless the initial frequency $f_c$ is sufficiently high. In the simulations and field experiments as presented in Section \ref{sec:experiment}, the signal processing algorithms do not adopt this assumption. The impact of the relative bandwidth will be discussed in the simulation section.

\subsection{Analysis on block coherence $\mu_I$ and $\mu_B$}
\label{subsec:uIuB}

According to the definitions  \eqref{Eq:intra-block coherence} and \eqref{Eq:inter-block coherence}, the block coherence $\mu_I$ and $\mu_B$ depend on the singular values of the matrix product $\bm{\Psi}_{q_1}^H \bm{\Psi}_{q_2}$ ($q_1 = q_2$ for $\mu_I$ and $q_1 \ne q_2$ for $\mu_B$). It is usually difficult to analyze singular values of a highly structured random matrix. Fortunately, the matrix products $\bm \Psi_{q_1}^H \bm \Psi_{q_2}$ are circulant matrices, as will be shown in the sequel, which enables us to obtain the closed-form expressions of their singular values with respect to the random carrier frequencies. Based on these analytical results, we then derive the statistical characters of the singular values and the consequent block coherence.

For the sake of clear presentation, we introduce the following notation. Let $\bm X$ and $\bm{X}_{q_1,q_2}$ be matrix products, particularly, $\bm{X} := \bm{\Psi}^H\bm{\Psi} \in \mathbb{C}^{MN\times MN}$ and $\bm{X}_{q_1,q_2} := \bm{\Psi}_{q_1}^H\bm{\Psi}_{q_2} \in \mathbb{C}^{M\times M}$, $q_1, q_2 \in \mathbb{N} $. From the definition of $\bm{\Psi}$ \eqref{eq:psi with blocks}, it can be verified that
\begin{align}
\label{eq:X}
    \bm{X} = \left[
               \begin{array}{cccc}
                 \bm{X}_{0,0} & \bm{X}_{0,1}  & \cdots & \bm{X}_{0,N-1} \\
                 \bm{X}_{1,0} & \bm{X}_{1,1}  & \cdots & \bm{X}_{1,N-1} \\
                 \vdots & \vdots & \ddots & \vdots \\
                 \bm{X}_{N-1,0} & \bm{X}_{N-1,1}  & \cdots & \bm{X}_{N-1,N-1} \\
               \end{array}
               \right],
\end{align}
indicating that $\bm{X}_{q_1,q_2}$ are the blocks of $\bm{X}$. 

Observing the definitions \eqref{Eq:intra-block coherence} and \eqref{Eq:inter-block coherence}, we define a variant of $\bm{X}_{q_1,q_2}$ as
\begin{align}
\label{Eq:Gram Matrix}
\overline{\bm{X}}_{q_1,q_2} :=\left\{
\begin{array}{lll}
\bm{X}_{q,q} - \bm{I}_M &  q_1 = q_2 = q, \\
\bm{X}_{q_1,q_2}    &  q_1 \neq q_2,
\end{array} \right.
\end{align}
so that $\mu_I$ and $\mu_B$ can be rewritten in an unified form as
\begin{equation}
\label{eq:uI_snorm}
    \mu_I = \max_{q \in \mathbb{N} } \left\| \overline{\bm{X}}_{q,q} \right\|_s,
\end{equation}
\begin{equation}
\label{eq:uB_snorm}
    \mu_B = \max_{\substack{q_1, q_2 \in \mathbb{ N}, \\ q_1 \neq q_2} } \left\| \overline{\bm{X}}_{q_1,q_2}\right\|_s,
\end{equation}
respectively.
As the spectral norm of a square matrix is highly related to its eigenvalues and singular values, we then denote by $\lambda$ and $\sigma$ the eigenvalue and singular value of a matrix, respectively. In particular, we use $\lambda_l$, $\lambda_m^{q_1,q_2}$ and $\bar{\lambda}_m^{q_1,q_2}$ ($\sigma_l$, $\sigma_m^{q_1,q_2}$ and $\bar{\sigma}_m^{q_1,q_2}$) to represent the $l$-th eigenvalue (singular value) of the matrix $\bm{X}$, and the $m$-th of $\bm{X}_{q_1,q_2}$ and $\overline{\bm{X}}_{q_1,q_2}$, respectively, $l \in \mathbb{L} := \{0, 1, \cdots, MN -1\}$, $m \in \mathbb{M}$.

With these notation, we now reveal in the subsequent Lemma \ref{lemma:1} that $\bm{X}_{q_1,q_2}$ is a circulant matrix, where each row is generated by moving the preceding row with one position to the right and wrapping around \cite{Davis1979}. This special structure will be later leveraged to derive the closed-form expressions of the eigenvalues $\lambda_m^{q_1,q_2}$. 
\begin{lemma}
\label{lemma:1}
The matrix $\bm{X}_{q_1,q_2}$, $ q_1,q_2 \in \mathbb{ N}$, is a circulant matrix.
\end{lemma}

\begin{proof}
See Appendix \ref{app:lemma circulant}.
\end{proof}

Since the eigenvalues of a circulant matrix are discrete Fourier transformation of its first row \cite{Chan1996}, we now derive the analytical expression of the eigenvalues $\lambda_m^{q_1,q_2}$ as stated in Lemma \ref{lemma:eigenvalues}.

\begin{lemma}
\label{lemma:eigenvalues}
The eigenvalues $\lambda_m^{q_1,q_2}$ can be expressed as
\begin{equation}
\label{Eq:Expression of lambda_m_q_1_q_2}
   \lambda_m^{q_1,q_2}  =  \frac{M}{N}\sum_{n=0}^{N-1}\zeta_{n,m}e^{j2\pi\frac{q_2 - q_1}{N}n},
\end{equation}
where $\zeta_{n,m} = \delta \left(C_n - m\right)$.
\end{lemma}
\begin{proof}
See Appendix \ref{app:lemma eigenvalues}.
\end{proof}

From \eqref{Eq:Expression of lambda_m_q_1_q_2}, we find that 1) the eigenvalue $\lambda_m^{q_1,q_2}$ is a random variable. The randomness comes from the randomly selected frequency code $C_n$. Recall that each frequency code obeys an i.i.d uniform distribution, i.e., $C_n \sim U \left(\mathbb{M} \right)$. We then have that $\zeta_{n,m}$ obeys a Bernoulli distribution  $\zeta_{n,m} \sim B\left(\frac{1}{M}\right)$, i.e.,
\begin{equation}
\label{eq:distribution}
\mathds{P}\left( \zeta_{n,m} = 0 \right) = 1 - \frac{1}{M}; \ \ \mathds{P}\left( \zeta_{n,m} = 1 \right) = \frac{1}{M}.
\end{equation}
And the random variables $\zeta_{n,m}, n \in \mathbb{N}$, are independent among each other for a fixed $m \in \mathbb{M}$.
This will be used later to derive the tail probability of the block coherence.

The result \eqref{Eq:Expression of lambda_m_q_1_q_2} also indicates that 2) the value of $\lambda_m^{q_1,q_2}$ depends on the difference $q_1$ and $q_2$, i.e. $\Delta q := q_2 - q_1$, and not on the particular values of block indices $q_1$ and $q_2$. Witnessing this, we replace the notation $\lambda_m^{q_1,q_2}$ with $\lambda_m^{\Delta q}$, $\Delta q \in \{\pm n\}_{n=0}^{N-1}$, for simplicity, which is given by
\begin{equation}
\label{Eq:Expression of lambda_m_delta q}
   \lambda_m^{\Delta q}  :=  \frac{M}{N}\sum_{n=0}^{N-1}\zeta_{n,m}e^{j2\pi\frac{\Delta q}{N}n}.
\end{equation}

We can further observe in \eqref{Eq:Expression of lambda_m_delta q} that 3) $\lambda_m^{\Delta q}$ has some particular conjugate-symmetric characters as stated in the following equations,
\begin{equation}
\label{Eq: conjugate symmetric structure1}
\lambda_m^{\Delta q} = \left(\lambda_m^{-\Delta q}\right)^*,
\end{equation}
\begin{equation}
\label{Eq: conjugate symmetric structure2}
\lambda_m^{\Delta q} = \lambda_m^{\Delta q \pm N},
\end{equation}
which imply that there are duplicated values in the magnitudes $\left\{ \left|\lambda_m^{\Delta q} \right| \right\}_{\Delta q \in \{\pm n\}_{n=0}^{N-1}}$, i.e.,
\begin{equation}
    \label{eq:duplicated magnitude}
    \left|\lambda_m^{\Delta q} \right| = \left|\lambda_m^{-\Delta q} \right| = \left|\lambda_m^{\Delta q \pm N} \right|, \Delta q \in \{\pm n\}_{n=0}^{N-1}.
\end{equation}




Given the expression of $\lambda_m^{\Delta q}$, we then derive ${\bar\lambda}_m^{q_1, q_2}$ and ${\bar\sigma}_m^{q_1, q_2}$. Invoking the definition \eqref{Eq:Gram Matrix} directly implies that ${\bar\lambda}_m^{q_1, q_2} = {\lambda}_m^{q_1, q_2} = {\lambda}_m^{\Delta q}$  for $q_1 \ne q_2$, and for $q_1 = q_2 = q$, ${\bar\lambda}_m^{q, q} = {\lambda}_m^{q, q} - 1 = {\lambda}_m^{0} - 1$. Thus, we find that ${\bar\lambda}_m^{q_1, q_2}$ also relies on the difference $\Delta q$ . Similarly, we define $\bar{\lambda}_m^{\Delta q} := {\bar\lambda}_m^{q_1, q_2}$ for $\Delta q = q_2 - q_1$, which can be expressed as
\begin{align}
\label{Eq: relation chi and chi_bar}
\bar{\lambda}_m^{\Delta q} =\left\{
\begin{array}{lll}
{\lambda}_m^{0} - 1   & \Delta q = 0, \\
{\lambda}_m^{\Delta q}&  \Delta q \neq 0.
\end{array} \right.
\end{align}

Since $\overline{\bm{X}}_{q_1,q_2}$ is also a circulant matrix, the singular values are given by the magnitudes of the eigenvalues\cite{Vybiral2011},
\begin{equation}
\label{Eq: relation between siagma and lambda}
\begin{split}
\bar{\sigma}_m^{q_1,q_2} = \left|\bar{\lambda}_m^{q_1,q_2}\right|,
\end{split}
\end{equation}
which indicates that the singular value $\bar{\sigma}_m^{q_1,q_2}$ also depends on $\Delta q$ and can be rewritten as $\bar{\sigma}_m^{\Delta q}$. Combing \eqref{Eq: relation chi and chi_bar} and \eqref{Eq: relation between siagma and lambda} yields
\begin{equation}
\label{eq:singular eigenvalue}
    \bar{\sigma}_m^{\Delta q} =\left\{
\begin{array}{lll}
\left|{\lambda}_m^{0} - 1\right|   & \Delta q = 0, \\
\left| {\lambda}_m^{\Delta q} \right| &  \Delta q \neq 0.
\end{array} \right.
\end{equation}
Invoking the fact that the spectral norm of a matrix equals the maximum singular values, we can rewrite the intra-block \eqref{eq:uI_snorm} and inter-block coherence \eqref{eq:uB_snorm} with respect to the singular value $\bar{\sigma}_m^{\Delta q}$ as
\begin{equation}
\label{eq:uI singular}
    \mu_I = \max_{q \in \mathbb{N} } \max_{m \in \mathbb{M} } \bar{\sigma}_m^{0} = \max_{m \in \mathbb{M} } \bar{\sigma}_m^{0},
\end{equation}
\begin{equation}
\label{eq:uB_singular}
    \mu_B = \max_{\Delta q \in \left\{ \pm n \right\}_{n =1}^{N-1} } \max_{m \in \mathbb{M} } \bar{\sigma}_m^{\Delta q},
\end{equation}
respectively.
Here, regarding \eqref{eq:uB_singular}, we note that among these $2N-2$ elements in $\bar{\sigma}_m^{\Delta q}$, $\Delta q \in \left\{ \pm n \right\}_{n =1}^{N-1} $, there are at most $\lfloor N/2\rfloor$ 
unique values, i.e., $\bar{\sigma}_m^{\Delta q}$, $\Delta q \in \left\{ n \right\}_{n =1}^{\lfloor N/2\rfloor} $. This is a consequence of applying \eqref{eq:singular eigenvalue} and \eqref{eq:duplicated magnitude}. Then, \eqref{eq:uB_singular} becomes
\begin{equation}
\label{eq:uB_singular compact}
    \mu_B = \max_{\Delta q \in \left\{ n \right\}_{n =1}^{\lfloor N/2\rfloor}} \max_{ m \in \mathbb{M}} \bar{\sigma}_m^{\Delta q}.
\end{equation}

Note that the singular values are random variables, and a bound of their tail probabilities are presented in Theorem \ref{theorem:tailsingular}. With the obtained results, we will later analyze the probabilistic characters of the block coherence $\mu_I$ and $\mu_B$.

\begin{theorem}
\label{theorem:tailsingular}
For $\epsilon \leq 1$ and $\Delta q \in \left\{ n \right\}_{n = 0}^{\lfloor N/2\rfloor}$, the singular values $\bar{\sigma}_m^{\Delta q}$ satisfy
\begin{equation}
\label{Eq: bound of bar sigma}
\mathds{P}\left( \bar{\sigma}_m^{\Delta q} > \sqrt{\frac{M-1}{N}} + \epsilon \right) < e^{-\frac{N}{4\left(M-1\right)}\epsilon^2}.
\end{equation}
\end{theorem}

\begin{proof}
See Appendix \ref{app:thm tail singular}.
\end{proof}

Given \eqref{Eq: bound of bar sigma}, we derive a probability bound on $\mu_I$ by applying the union bound to \eqref{eq:uI singular}. In particular, let $c_1 = \sqrt{\frac{M-1}{N}} + \epsilon$, and we have
\begin{equation}
\label{Eq:uI bound}
\begin{split}
\mathds{P}\left(\mu_I > c_1\right) & = \mathds{P}\left(\max_{m \in \mathbb{M}}{\bar{\sigma}_m^{0}} > c_1\right)\\
      & \leq \sum_{m \in \mathbb{M} }\mathds{P}\left({\bar{\sigma}_m^{0}} > c_1\right)\\
      & < M e^{-\frac{N}{4\left(M-1\right)}\epsilon^2}\\
      & = M e^{\tiny{-\frac{\left(\sqrt{N}c_1 - \sqrt{M-1}\right)^2}{4\left(M-1\right)}}}.
\end{split}
\end{equation}
Regarding with $\mu_B$ in \eqref{eq:uB_singular compact}, we follow the same technique in \eqref{Eq:uI bound}, and obtain the subsequent bound
\begin{equation}
\label{Eq:uB bound}
\mathds{P}\left(\mu_B > c_2\right)< M \lfloor N/2 \rfloor e^{\tiny{-\frac{\left(\sqrt{N}c_2 - \sqrt{M-1}\right)^2}{4\left(M-1\right)}}}.
\end{equation}

Probability bounds \eqref{Eq:uI bound} and \eqref{Eq:uB bound} characterize the block coherence of the observation matrix $\bm \Psi$. To establish a sufficient condition for exact recovery that uses $\bm \Psi$, we calculate its spectral norm $\left\| \bm \Psi \right\|_s$ in the next subsection.
\subsection{Derivation of \texorpdfstring{$\left\|\bm \Psi \right\|_s$}{textpsinorm}}
\label{subsec:psis}

Despite of the randomness, we find that $\bm \Psi$ has a determinate spectral norm $\left\|\bm \Psi \right\|_s = \sqrt{M}$. To reveal this, we start by analyzing the structure of the Gram matrix $\bm X = \bm \Psi^H \bm \Psi$, since the singular values of a matrix correspond to the eigenvalues of its Gram matrix \cite{zhang2017matrix}. Particularly, the $l$-th singular value of $\bm \Psi$ satisfies
\begin{equation}
\label{eq:sqrt eigenvalue}
    \sigma_l(\bm \Psi) = \sqrt{\lambda_l(\bm X)} = \sqrt{\lambda_l}, \quad l \in \mathbb{L}.
\end{equation}
Leveraging a particular structure (as will be stated in Lemma \ref{lemma:block circulant}), we then derive the analytical form of the eigenvalues $\lambda_l$, which completes the calculation of $\left\|\bm \Psi \right\|_s$.

From \eqref{eq:X}, we find that 1) $\bm X$ has circulant blocks, since each block $\bm X_{q_1,q_2}$ is a circulant matrix as revealed in Lemma \ref{lemma:1}. Besides this, $\bm X$ has an additional structure. As we will prove later in the Appendix \ref{app:block circulant}, 
each (block) row of $\bm X$ is a right cyclic shift of the row above it, i.e., for $q_1 = 1, 2, \dots, N-1$,
\begin{align}
\label{Eq: block circulant}
\bm{X}_{q_1,q_2} =\left\{
\begin{array}{lll}
\bm{X}_{q_1 - 1, q_2 - 1}      &   q_2 = 1, 2, \dots, N-1, \\
\bm{X}_{q_1 - 1, N-1}  &   q_2 = 0,
\end{array} \right.
\end{align}
or equivalently, for $q_1,q_2 \in \mathbb{N}$,
\begin{align}
\label{Eq: block circulant2}
\bm{X}_{q_1,q_2} =\left\{
\begin{array}{lll}
\bm{X}_{0,q_2-q_1}      &   q_2 \geq q_1, \\
\bm{X}_{0,q_2-q_1 + N}  &   q_2 < q_1,
\end{array} \right.
\end{align}
which indicates that 2) $\bm X$ is a block criculant matrix. Combining 1) and 2) implies the following Lemma \ref{lemma:block circulant}.

\begin{lemma}
\label{lemma:block circulant}
The matrix $\bm{X}$ is a block circulant matrix with circulant blocks.
\end{lemma}
\begin{proof}
This is a simple consequence of the previous discussion.
\end{proof}

For a matrix with such structure, its eigenvalues are given by the eigenvalues ($\lambda_m^{n}$, $m \in \mathbb{M}$) of the circulant blocks $\bm X_{0,n}$, $n \in \mathbb{N}$ \cite[Thm 5.8.1]{Davis1979}. In particular, for $l = qM + m$, $q \in \mathbb{N}$, we have the following
\begin{equation}
\label{Eq:eigenvalue expression of X}
    \lambda_l = \sum_{n=0}^{N-1}e^{j2\pi\frac{qn}{N}}\lambda_m^{n}.
\end{equation}

Given the eigenvalues of $\bm X$, we then obtain the singular values of $\bm \Psi$ using \eqref{eq:sqrt eigenvalue}. Finding the maximum of these singular values yields $\left\| \bm \Psi \right\|_s$, as stated in the following corollary.

\begin{corollary}
\label{cor:psi_s}
The spectral norm of $\bm{\Psi}$ is given by
\begin{equation}
\label{Eq:spectral norm epxpression}
    \|\bm{\Psi}\|_s = \sqrt{M}.
\end{equation}
\end{corollary}
\begin{proof}
See Appendix \ref{sec:proof of spectral norm}.
\end{proof}

Using probability bounds on $\mu_I$ and $\mu_B$ together with $\left\|\bm \Psi\right\|_s$, we are now ready to derive a unique recovery condition in the subsequent subsection.

\subsection{Unique recovery condition}
\label{subsec:bound}

Based on the condition \eqref{Eq: original sufficient condition} in Theorem \ref{theorem:origin}, we develop a requirement on $M$, $N$ (radar parameters) and $K$ (the block sparsity, i.e., the number of extended targets), under which the observation matrix $\bm \Psi$ meets the condition \eqref{Eq: original sufficient condition} with high probability. We state the main result in the following theorem.
\begin{theorem}
\label{theorem:unique recovery}
For any constant $\epsilon > 0$ and a sufficiently large $N$, the inequality \eqref{Eq: original sufficient condition} holds with a probability at least $1 - \epsilon$ when the block sparsity satisfies
\begin{equation}
\label{Eq:Block CS sparsity constrain}
\begin{split}
K \leq  \frac{N\left(\frac{1}{8} - \delta_1 - \delta_2\right)^2}{81{M\log MN \left(1 + 2\delta_2/3\right)}},
\end{split}
\end{equation}
where $\delta_1 := 24\sqrt{\frac{M-1}{N}}\log{MN}\left(2\sqrt{\log{MN}-\log\epsilon} + 1\right)$ and $\delta_2 := \frac{3}{2}\sqrt{\frac{M-1}{N}}\left(2\sqrt{\log{2M}-\log\epsilon} + 1\right)$ are small constants for large $N$.
\end{theorem}
\begin{proof}
See Appendix \ref{sec:appendices 1}.
\end{proof}

Theorem \ref{theorem:unique recovery} reveals that the observation matrix of RSFR satisfies \eqref{Eq: original sufficient condition} with high probability (with respect to the random selection of carrier frequencies) if the number of the extended targets, i.e., the block number, is in the order of $K = O\left({\frac{N}{M\log MN}}\right)$.
In this case, according to Theorem \ref{theorem:origin}, we obtain average-case guarantees for range-velocity reconstruction in RSFR.
In terms of the number of scattering points, the scale becomes $KM = O\left({\frac{N}{\log MN}}\right)$ since each block contains $M$ elements.
Comparing this result with the previous bound that was built on canonical (i.e., non-block) sparse recovery \cite{Huang2018TSP}, i.e.,  $KM = O\left(\sqrt{\frac{N}{\log MN}}\right)$, makes us optimistic to use block sparse recovery in RSFR. In the ensuing section, both synthetic and field experiments are executed, results of which demonstrate that block sparse recovery leads to better performance on range-velocity reconstruction than the non-block counterparts.
\section{Experimental Results}
\label{sec:experiment}
In this section, both simulated and measured data are provided to test the performance of the block and non-block sparse recovery algorithms.
We consider noiseless and noisy scenarios. In noiseless cases, we use the the mixed $\ell_{2,1}$ norm minimization and the Block-OMP as examples of block sparse recovery algorithms, the $\ell_{1}$ norm minimization and OMP as the counterparts of non-block algorithms for comparison, respectively. In noisy cases, the $\ell_{1}$ and  mixed $\ell_{2,1}$ norm minimization become Lasso and Block-Lasso, respectively. When we deal with measured data, we apply matched filter additionally, which simply reconstructs the range-Doppler parameters as ${\hat{\bm x}} = \bm \Psi^H \bm y$.
The results demonstrate the effectiveness of block sparse recovery and imply its superiority over conventional sparse recovery. Specifically, the simulation results are presented in Subsection \ref{sec:exp:simu}, in which we focus on two aspects: 1) the statistical property of the observation matrix $\bm \Psi$ including $\mu_I, \mu_B$ and $\|\bm{\Psi}\|_s$, and 2) the reconstruction performance in both the noiseless and noisy scenarios.
Then, the measured-data results are provided in Subsection \ref{sec:exp:field data}, in which the reconstruction performance of both air and surface target scenarios is tested. 

\subsection{Simulation Results}
\label{sec:exp:simu}
In this subsection, three simulation experiments are conducted for different considerations. In the first experiment, we focus on the block coherence and the spectral norm of the observation matrix $\bm{\Psi}$, and study the impact of the relative bandwidth on them. In the second and third experiments, we discuss the range-velocity reconstruction performance in the noiseless and noisy scenarios, respectively.

In the first experiment, we study the Complementary Cumulative Distribution Functions (CCDFs) of $\mu_I$, $\mu_B$ and $\|\bm{\Psi}\|_s$. We set the pulse number and the frequency point number as $N = 32$ and $M = 4$, respectively. Different Relative Bandwidths (RB) are simulated, which is defined as $\text{RB} = \frac{M\Delta f}{f_c}$. The observation matrix $\bm{\Psi}$ is generated according to \eqref{Eq: observation vector}.
In the cases when RB$=0.01$ or $0.1$, the assumption $\xi_n = 1$ does not apply. And we denote by "RB$=0$" when we generate $\bm{\Psi}$ under the assumption $\xi_n = 1$.
The results are shown in Fig. \ref{fig:simu: muI}-\ref{fig:simu: spec norm}, in which the theoretical bounds of $\mu_I$ and $\mu_B$, presented in \eqref{Eq:uI bound} and \eqref{Eq:uB bound}, respectively, are also shown for comparison.
As expected, when the assumption $\xi_n = 1$ applies, the CCDFs of $\mu_I$ and $\mu_B$ are bounded by \eqref{Eq:uI bound} and \eqref{Eq:uB bound}, respectively; and $\|\bm{\Psi}\|_s = \sqrt{M}$.
When $\xi_n = 1$ does not hold, the CCDF of $\mu_I$ does not change, which can also be deduced from the fact that the definition of $\mu_I$ \eqref{eq:uI_snorm} is irrelevant with $\xi_n$.
However, in this situation, CCDFs of $\mu_B$ and $\|\bm{\Psi}\|_s$ change.
In the tested scenarios, $\mu_B$ and $\|\bm{\Psi}\|_s$ tend to take values slightly larger than those when we assume $\xi_n = 1$. Changes of $\mu_B$ and $\|\bm{\Psi}\|_s$ may affect reconstruction performance in RSFR applications and we leave the theoretical analysis for future investigation. However, as indicated in the following experiments, the block sparse recovery algorithms still enjoy satisfactory reconstruction performance though the assumption $\xi_n = 1$ does not apply.

\begin{figure}[htbp]
	\begin{minipage}[b]{1.0\linewidth}
		\centering
		\centerline{\includegraphics[width=7cm]{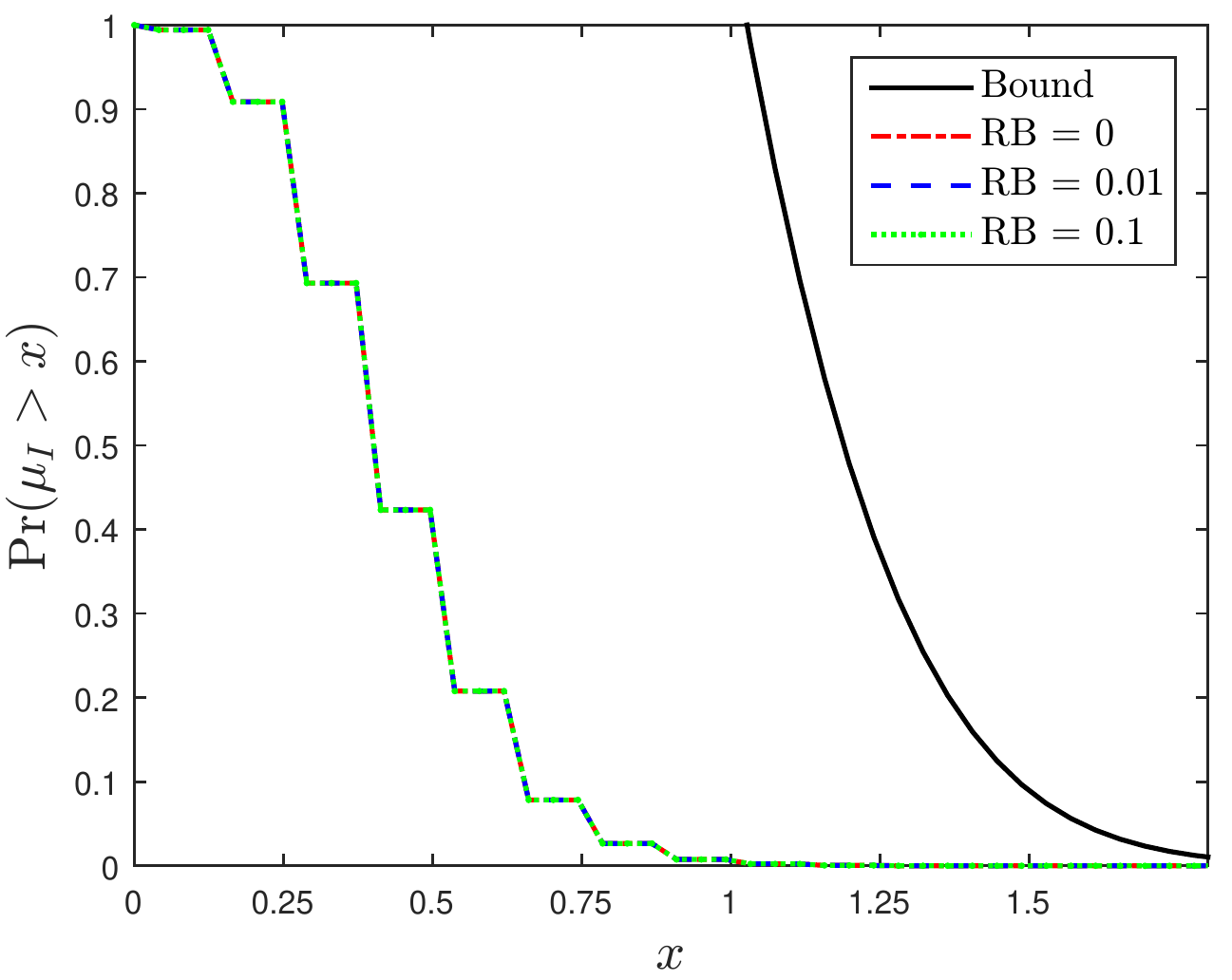}}
	\end{minipage}
	\caption{CCDFs of $\mu_I$ with $N=32$ and $M=4$.}
	\label{fig:simu: muI}
	\begin{minipage}[b]{1.0\linewidth}
		\centering
		\centerline{\includegraphics[width=7cm]{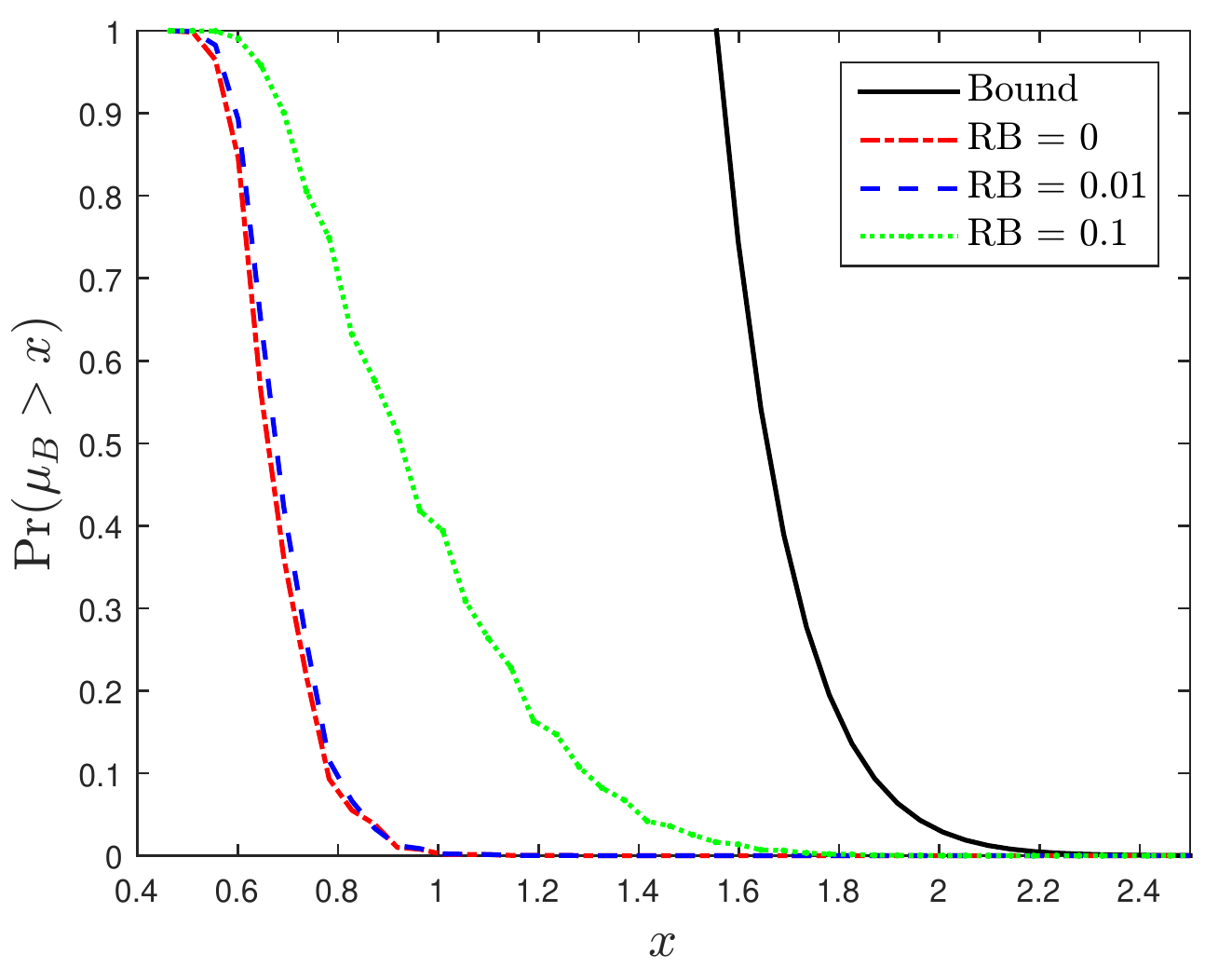}}
	\end{minipage}
	\caption{CCDFs of $\mu_B$ with $N=32$ and $M=4$.}
	\label{fig:simu: muB}
	\begin{minipage}[b]{1.0\linewidth}
		\centering
		\centerline{\includegraphics[width=7cm]{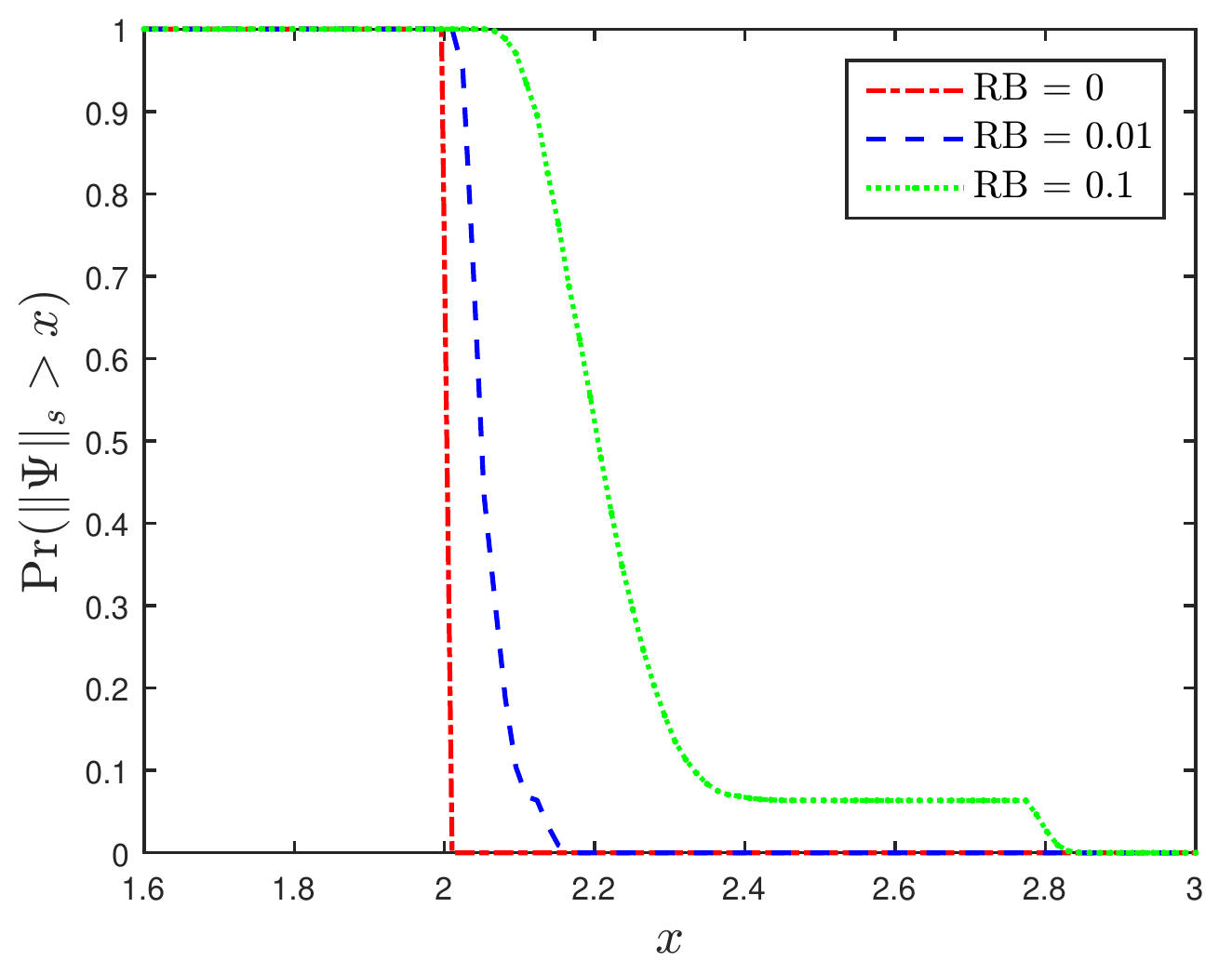}}
	\end{minipage}
	\caption{CCDFs of $\|\bm{\Psi}\|_s$ with $N=32$ and $M=4$.}
	\label{fig:simu: spec norm}
\end{figure}
\begin{figure}[htbp]
\begin{minipage}[b]{1.0\linewidth}
 \centering
 \centerline{\includegraphics[width=7.5cm]{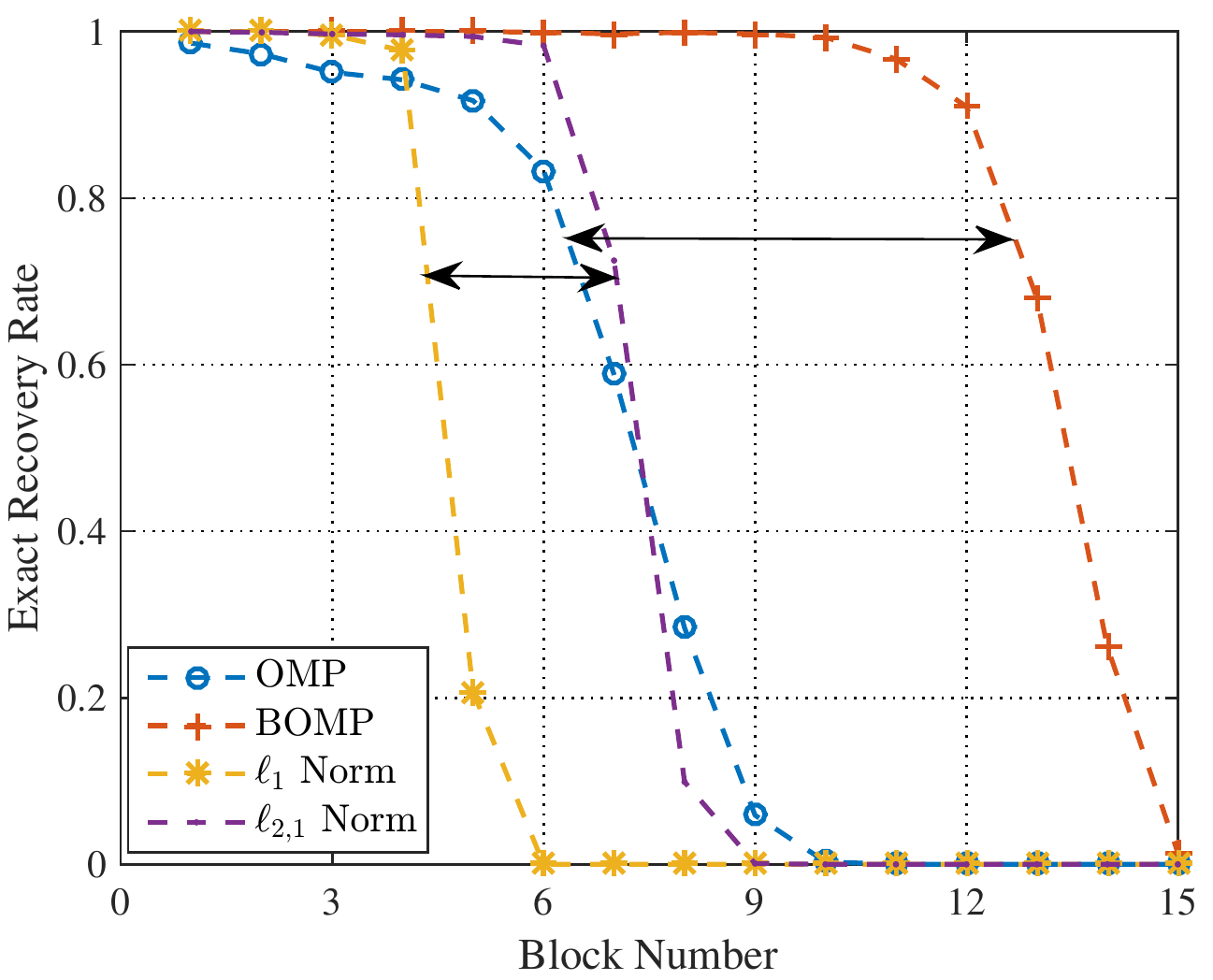}}
\end{minipage}
\caption{The exact recovery ratio comparison in the noiseless case.}
\label{fig:simu: Noiseless ERR}
\end{figure}
\begin{figure}[htbp]
\begin{minipage}[b]{1.0\linewidth}
 \centering
 \centerline{\includegraphics[width=8.6cm]{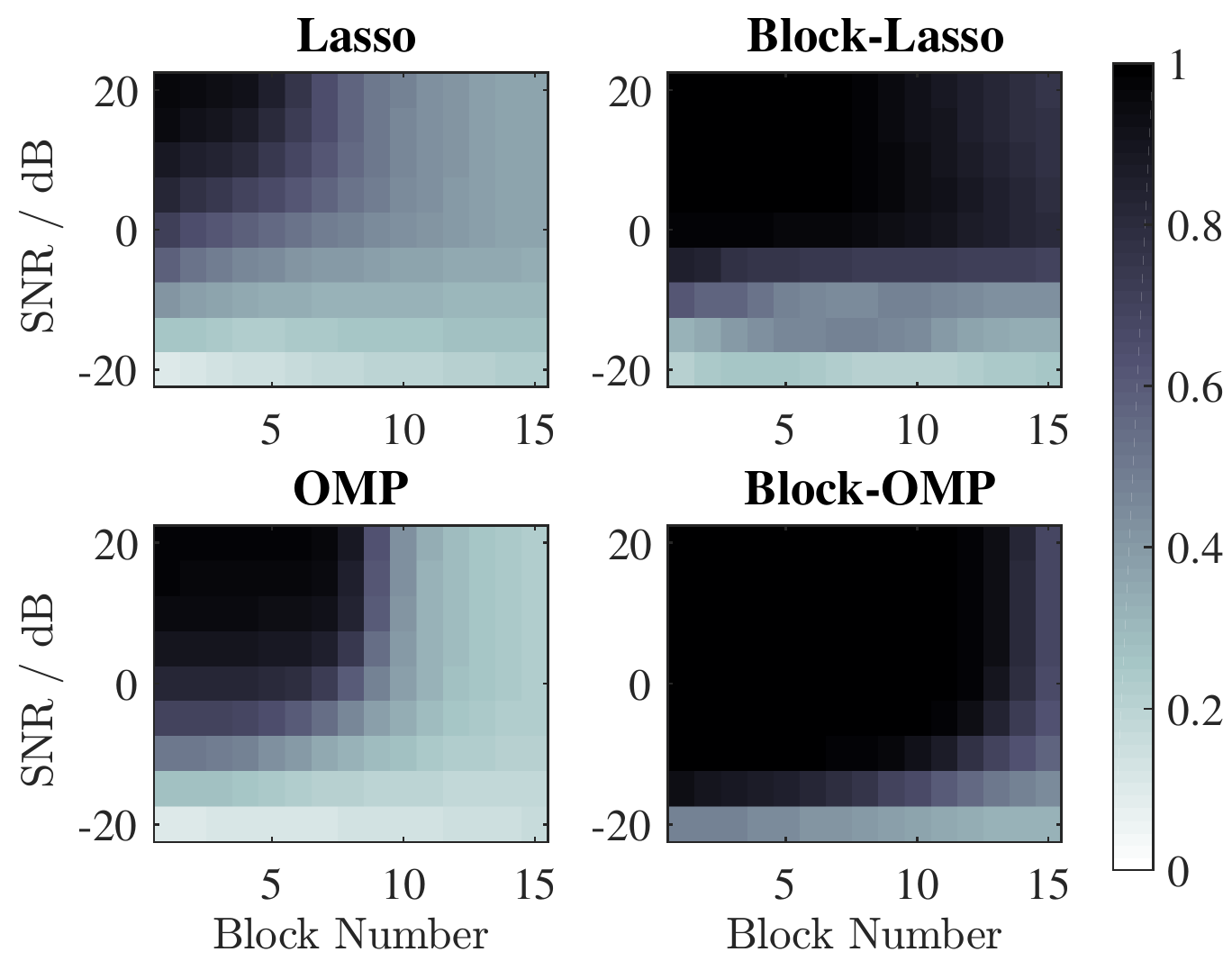}}
\end{minipage}
\caption{The hit rate of different algorithms in the noisy case.}
\label{fig:simu: phase transtion}
\end{figure}

In the next two experiments, we verify the robustness of block sparse recovery algorithms and their superiority in comparison with conventional sparse recovery algorithms.
The RSFR works with the parameters as follows: $f_c = 9$GHz, $T_r = 20\mu$s, $\Delta f = 30$MHz, $M = 8$ and $N =128$. Multiple extended targets are simulated. And each one consists of $P_k = 8$ scatterers distributed in cluster along range. The velocity frequencies of the targets $\{f_{v_k}\}$ are randomly uniformly selected from the grid points $\{q/N\}_{q \in \mathbb{N}}$ and each scatterer has a random scattering coefficient obeying a complex Gaussian distribution as $\gamma_{ki}\sim \mathcal{CN}\left(0,1\right)$. When we run OMP and Block-OMP, the sparsity and block sparsity are set as the true numbers of scatterers and targets, respectively. In the $\ell_1$ and mixed $\ell_{2,1}$ norm minimization, an entry $\hat{x}$ of $\hat{\bm x}$ is regarded as nonzero and its index is contained in the support set when the magnitude exceeds a rather low threshold (i.e., $|\hat{x}|>10^{-5}$). For Lasso and Block-Lasso used in the noisy scenario, the support sets are identified by seeking for elements in $\hat{\bm{x}}$ with $\sum_k P_k$ largest magnitudes, where $\sum_k P_k$ denotes the true scatterer number.

\begin{figure*}[htbp]
	\begin{minipage}[b]{1.0\linewidth}
		\centering
		\centerline{\includegraphics[width=18.0cm]{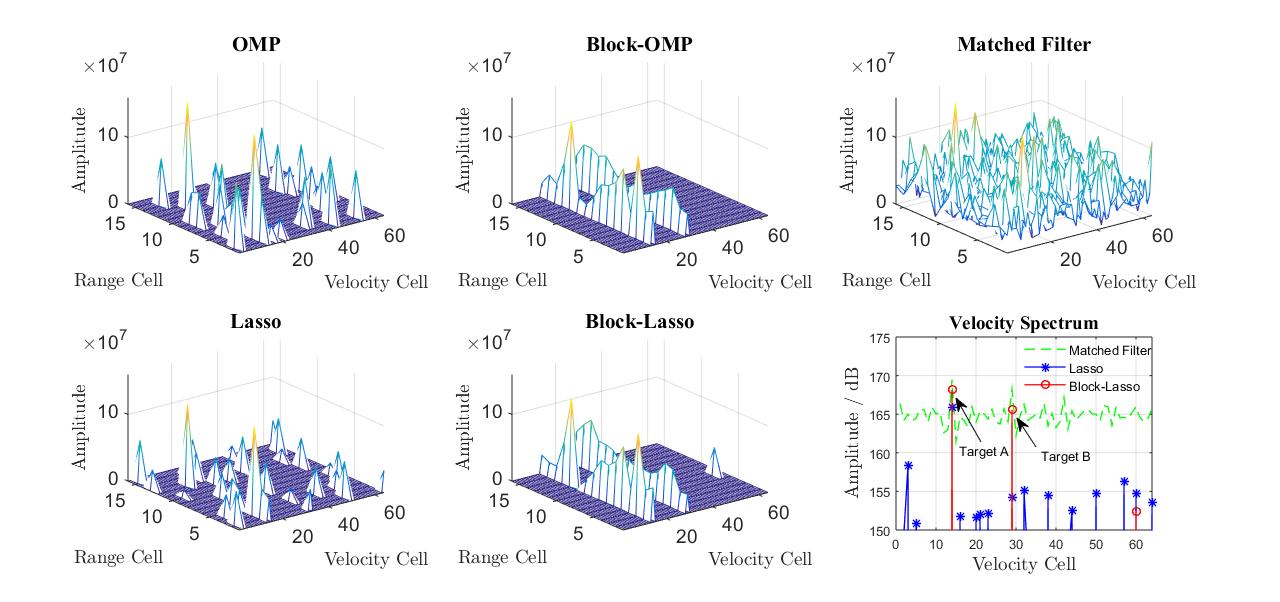}}
	\end{minipage}
	\caption{The reconstructed air targets using different methods.}
	\label{fig:real data: Plane}
\end{figure*}
The second experiment considers the noiseless case. We repeat the simulation independently for $1000$ times under different target numbers (i.e., block numbers). Exact recovery rate is introduced to evaluate the recovery performance. An exact recovery is proclaimed when the recovered support set exactly matches the ground truth and the exact recovery rate is defined as the ratio of exact recovery times to the overall simulation times. 
Results are shown in Fig. \ref{fig:simu: Noiseless ERR}. We find that block sparse recovery algorithms reconstruct larger numbers of blocks than the non-block counterparts.
The performance improvement is indicated by double-headed arrows, which reveals the superiority of block sparse recovery over the conventional sparse recovery.

The third experiment is designed for noisy scenarios. Due to the existence of noise, exact recovery is hard to obtain. Instead, we introduce hit rate for performance evaluation, which is defined as the ratio of the number of correctly recovered nonzero entries to the total number of nonzero entries. Different block numbers and Signal-to-Noise Ratio (SNR) are simulated. The SNR is defined as $\text{SNR} = 10\log_{10}{\frac{1}{\sigma^2}}$. Hit rate results are shown in Fig. \ref{fig:simu: phase transtion}, in which each square is obtained by averaging $1000$ Monte Carlo trials. Larger area of the dark color part represents better performance, which also reveals the superiority of block sparse recovery over the conventional sparse recovery.

\subsection{Measured Data Results}
\label{sec:exp:field data}
In this subsection, measured data is used to demonstrate the effectiveness of the block sparse recovery methods. Radar data is obtained in two different target scenarios. One is an air-target scenario, where there are two targets but no clutter, while the other is a surface-target scenario with only one target and serious clutter. For OMP and Block-OMP, the sparsity and block sparsity are set to $2M$ and $2$, respectively.
In Lasso and Block-Lasso, we identify $2M$ elements of $\hat{\bm x}$ with largest magnitudes and apply least squares to these elements to refine the estimates of scattering coefficients.
Finally recovered $\hat{\bm{x}}$ is reshaped into matrix form $\bm \Gamma$ according to the definition \eqref{eq:x}, and the magnitudes of $\bm \Gamma$ will be shown in Fig. \ref{fig:real data: Plane} and Fig. \ref{fig:real data: boat} corresponding to the range-velocity plane. Note that each range cell and velocity cell correspond to a grid point of range frequency and velocity frequency, respectively.
In addition, the velocity spectrum is introduced to evaluate the velocity estimation performance of extended targets. Specifically, the $q$-th entry of the velocity spectrum $\bm{\nu}_s$ is defined as the $\ell_2$ norm of the estimated HRR profile in the $q$-th velocity cell, i.e.,
\begin{equation}
    \left[\bm{\nu}_s\right]_q = \| \hat{\bm{x}}_q\|_2, \  q \in \mathbb{N},
\end{equation}
where $\hat{\bm{x}}_q$ is the $q$-th block of $\hat{\bm{x}}$.

\subsubsection{Air Target Recovery} Here, we provide experimental results with real measured RSFR data from civil aircraft in air. During the observing time, the radar transmits $N = 64$ pulses in a CPI, whose carrier frequencies cover $M = 16$ frequency points with the frequency step size $\Delta f = 50$MHz. Two civil aircraft are flying away from the radar with relative velocities $35$m$/$s (Target A) and $45$m$/$s (Target B), respectively.
Echoes from these two aircraft are measured individually with identical radar waveform.
We add the echoes of the two aircraft to generate a two-target scenario. For the measured two-target-scenario data, the SNR is high, and extra complex Gaussian noise is added to lower the SNR to approximately $10\text{dB}$. 
Magnitudes of reconstructed $\bm \Gamma$ are shown in Fig. \ref{fig:real data: Plane}. From Fig. \ref{fig:real data: Plane}, we observe that the matched filter result has a high sidelobe pedestal: only a few strong scatterers of Target A are distinct but many other scatterers are submerged by the sidelobe pedestal. Non-block sparse recovery algorithms, i.e., OMP and Lasso, extract more strong scatterers of Target A than matched filter, but lead to many spurious peaks. In the resuls of Block-OMP and Block-Lasso, both the HRR profiles of Target A and Target B are extracted and only one spurious peak with weak magnitude appears in the Block-Lasso result. These results demonstrate the validness of the block sparse recovery for extended targets and the advantage over the non-block counterparts.

The results of velocity spectrum are also shown in Fig. \ref{fig:real data: Plane}. To make the results clear, we only show the results of matched filter, Lasso and Block-Lasso because the results of the OMP
\begin{figure*}[htbp]
	\begin{minipage}[b]{1.0\linewidth}
		\centering
		\centerline{\includegraphics[width=18.0cm]{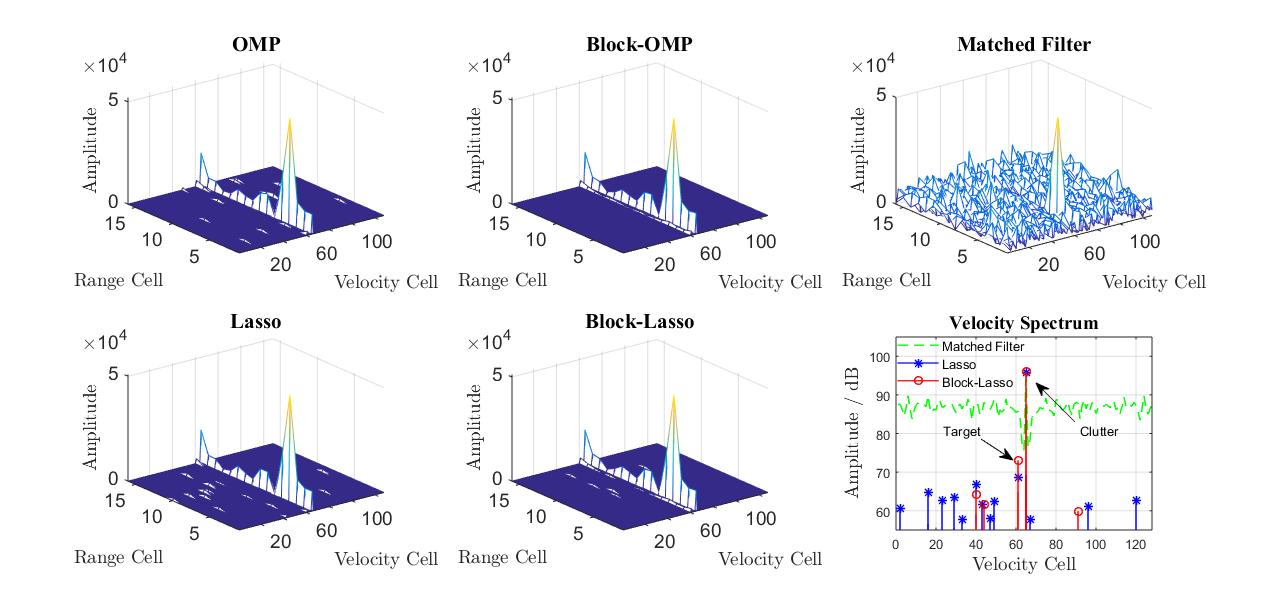}}
	\end{minipage}
	\caption{The reconstructed surface target and clutter using different methods.}
	\label{fig:real data: boat}
\end{figure*}
\begin{figure}[htbp]
	\begin{minipage}[b]{1.0\linewidth}
		\centering
		\centerline{\includegraphics[width=7.0cm]{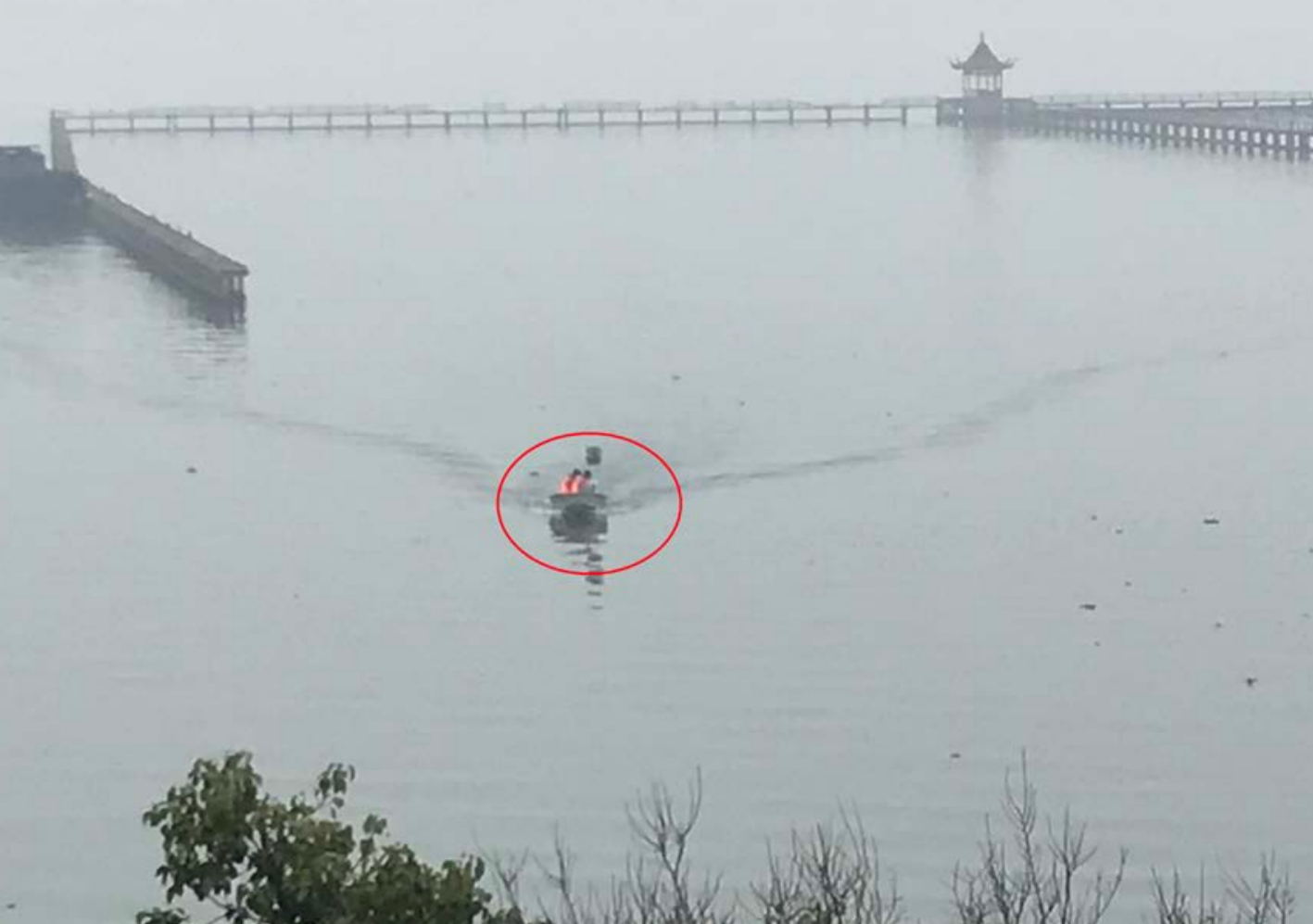}}
	\end{minipage}
	\caption{The photograph of the surface target scene.}
	\label{fig:real data: surface scenario}
\end{figure}
and Block-OMP are similar with those of Lasso and Block-Lasso, respectively. As one can find, all the three spectrums have two peaks corresponding to the true velocities of the targets. 
However, in matched filter and Lasso, the peaks are not distinct and many spurious peaks have magnitudes close to the true peaks.
In contrast, Block-Lasso results in sharp velocity spectrum with only one spurious peak approximately $13$dB weaker than those of the targets.

\subsubsection{Surface Target and Clutter Recovery}
We also use a RSFR to measure a boat moving in a lake. In this field experiment, the radar is amounted about $10$m above the lake surface and the boat is moving with a relative velocity about $2\text{m}/\text{s}$. The target scene is presented in Fig. \ref{fig:real data: surface scenario}, which is a photograph taken from the radar site. The RSFR is configured with parameters: $N = 128$, $M = 16$ and $\Delta f = 32$MHz. Due to the backscatter from the lake surface and other static objects, there are serious clutter in the received echoes and the clutter-to-signal ratio is approximately $25$dB.

We demonstrate the reconstructed $\bm \Gamma$ in Fig. \ref{fig:real data: boat}. Due to the dominant intensities of the clutter, the HHR profiles of the clutter are recovered identically by all the five methods tested. In the result of matched filter, however, the target is completely submerged in the sidelobe pedestal of the clutter. Though the non-block OMP and Lasso indicate the true velocity of the boat, many spurious peaks appear in their reconstructed $\bm \Gamma$. Some of these spurious peaks have even higher amplitudes than the true target, which may lead to false alarm. While in the results of Block-OMP and Block-Lasso, both HRR profiles of the target and clutter are reconstructed and there are only a few spurious peaks appearing in the Block-Lasso result with inferior amplitudes in comparison with those of the dominant scatterers in the target. These measured data results demonstrate the effectiveness of the block sparse recovery algorithms for range-Doppler reconstruction of extended targets.
Since similar conclusions can be drawn from the velocity spectrum results, we omit the detailed discussions.
\section{Conclusion}
\label{sec:conclusion}

In this paper, we consider the range-velocity reconstruction of extended targets in RSFRs. By exploiting the natural block sparsity of the extended targets, we introduce the block sparse recovery. We analyze the block coherence and the spectral norm of the observation matrix, and then establish a bound on the radar parameters, i.e., $KM = O\left({\frac{N}{\log MN}}\right)$, under which the range-velocity reconstruction is guaranteed. The obtained theoretical bound relaxes the previous constrain $KM = O\left(\sqrt{\frac{N}{\log MN}}\right)$, which is based on non-block sparse recovery. Both simulated and measured data results demonstrate the effectiveness of the block sparse recovery algorithms used for RSFRs and their superiority over the non-block counterparts.

\appendices
\section{Proof of Lemma \ref{lemma:1}}
\label{app:lemma circulant}

According to the definition of a circulant matrix \cite[Chapter 3]{Davis1979}, it is equivalent to prove that for $p_1,p_2 \in \mathbb{M}$, the $(p_1,p_2)$-th element satisfies
\begin{equation}
\label{eq:circulantXq1q2}
\left[\bm{X}_{q_1,q_2}\right]_{p_1,p_2}  = \left\{
\begin{array}{ll}
\left[\bm{X}_{q_1,q_2}\right]_{0,p_2-p_1}   & p_2 \geq p_1, \\
\left[\bm{X}_{q_1,q_2}\right]_{0,p_2-p_1 + M}&  p_2 < p_1.
\end{array} \right.
\end{equation}

Substituting the definitions $\bm{X}_{q_1,q_2} = \bm{\Psi}_{q_1}^H\bm{\Psi}_{q_2}$ and \eqref{Eq: simplified observation vector} 
into the left hand side of \eqref{eq:circulantXq1q2}, we have
\begin{equation}
\label{eq:entryXq1q2}
\begin{split}
\left[\bm{X}_{q_1,q_2}\right]_{p_1,p_2}  &= \bm{\psi}_{p_1,q_1}^H\bm{\psi}_{p_2,q_2} \\
&= \frac{1}{N}\sum_{n=0}^{N-1}e^{j2\pi\frac{p_2-p_1}{M}C_n + j2\pi\frac{q_2-q_1}{N}n}.
 \end{split}
\end{equation}
Noticing the phase term
\begin{equation}
    \begin{split}
        e^{j2\pi\frac{p_2-p_1}{M}C_n} &= e^{j2\pi\frac{(p_2-p_1)-0}{M}C_n} \\
        &= e^{j2\pi\frac{(p_2-p_1+M)-0}{M}C_n},
    \end{split}
\end{equation}
we have
$\left[\bm{X}_{q_1,q_2}\right]_{p_1,p_2} = \left[\bm{X}_{q_1,q_2}\right]_{0,p_2-p_1}$ and $\left[\bm{X}_{q_1,q_2}\right]_{p_1,p_2} = \left[\bm{X}_{q_1,q_2}\right]_{0,p_2-p_1+M}$, as long as $p_2-p_1$ and $p_2 - p_1 +M$ belong to $\mathbb{M}$, respectively, which proves \eqref{eq:circulantXq1q2}.
Therefore $\bm{X}_{q_1,q_2}$ is a circulant matrix.


\section{Proof of Lemma \ref{lemma:eigenvalues}}
\label{app:lemma eigenvalues}

Remember that $\bm{X}_{q_1,q_2}$ is a circulant matrix according to Lemma \ref{lemma:1}, whose eigenvalues are discrete Fourier transformation of its first row \cite{Chan1996}. For brevity, we let $\chi_p := \left[\bm{X}_{q_1,q_2}\right]_{0,p} \in \mathbb{C}$ be the $p$-th element of the first row in matrix $\bm{X}_{q_1,q_2}$.
From \eqref{eq:entryXq1q2}, $\chi_p$ is given by
\begin{equation}
\begin{split}
\chi_p  &= \frac{1}{N}\sum_{n=0}^{N-1}e^{j2\pi\frac{p}{M}C_n + j2\pi\frac{q_2-q_1}{N}n}.
\end{split}
\end{equation}
By applying discrete Fourier transformation to $[\chi_0, \chi_1,\dots,\chi_{M-1}]^T$, we obtain the eigenvalue as
\begin{equation}
\label{Eq:lambda mq1q2}
\begin{split}
\lambda_m^{q_1,q_2}  & = \sum_{p=0}^{M-1}{\chi_p e^{-j2\pi\frac{m}{M}p}}\\
 & = \frac{1}{N}\sum_{p=0}^{M-1}\sum_{n=0}^{N-1}{e^{j2\pi\frac{q_2 - q_1}{N}n + j2\pi\frac{C_n - m}{M}p}}\\
 & = \frac{1}{N}\sum_{n=0}^{N-1}\left(e^{j2\pi\frac{q_2 - q_1}{N}n}\sum_{p=0}^{M-1}{e^{j2\pi\frac{C_n - m}{M}p}}\right) \\
 & =  \frac{1}{N}\sum_{n=0}^{N-1} e^{j2\pi\frac{q_2 - q_1}{N}n} \cdot M\delta \left( C_n - m\right) \\
 & =  \frac{M}{N}\sum_{n=0}^{N-1}e^{j2\pi\frac{q_2 - q_1}{N}n}\cdot \zeta_{n,m} ,
 \end{split}
\end{equation}
where $\zeta_{n,m} := \delta \left(C_n - m\right) \in \{0,1\}$, completing the proof.

\section{Proof of Theorem \ref{theorem:tailsingular}}
\label{app:thm tail singular}
We  prove the theorem in the cases $\Delta q \neq 0$ and $\Delta q = 0$ individually.

In the case of $\Delta q \neq 0$, from \eqref{Eq:Expression of lambda_m_delta q} and \eqref{eq:singular eigenvalue}, we have
\begin{equation}
    \bar{\sigma}_m^{\Delta q} = \left| \frac{M}{N}\sum_{n=0}^{N-1} e^{j2\pi\frac{\Delta q}{N}n}\cdot \zeta_{n,m} \right|.
\end{equation}
Let $J_n = \frac{M\zeta_{n,m} - 1}{N}e^{j2\pi\frac{\Delta q}{N}n}$, we can verify that
\begin{equation*}
    \begin{split}
        \left| \sum_n^{N-1} J_n \right| =
        \left| \frac{M}{N}\sum_{n=0}^{N-1} e^{j2\pi\frac{\Delta q}{N}n} \zeta_{n,m} - \frac{1}{N}\sum_{n=0}^{N-1} e^{j2\pi\frac{\Delta q}{N}n}\right|
        = \bar{\sigma}_m^{\Delta q},
    \end{split}
\end{equation*}
by invoking the fact that $\sum_{n=0}^{N-1} e^{j2\pi\frac{\Delta q}{N}n} = 0$ for $\Delta q \ne 0$. We can also find that $J_n$ is independent from each other, which comes from the independence of $\zeta_{n,m}$ with respect to $n$. According to the distribution of $\zeta_{n,m}$ as mentioned in \eqref{eq:distribution},  we have  $\mathds{E}\left[{J}_n\right] = {0}$.
Then (\ref{Eq: bound of bar sigma}) is a direct consequence of the Bernstein inequality \cite[Thm. 12]{5714248} with $V$ and $\epsilon$ given by
\begin{equation}
\begin{split}
V & = \sum_{n=0}^{N-1} \mathds{E} \left[\left|{J}_n\right|^2\right]\\
 & = \frac{M^2}{N^2}\sum_{n=0}^{N-1}\mathds{E}\left[\left(\zeta_{n,m} - 1/M\right)^2\right] \\
 &= \frac{M-1}{N},
\end{split}
\end{equation}
and
\begin{equation}
\begin{split}
\epsilon \leq V/\max_{n\in \mathbb{N}} {|J_n|} = 1,
\end{split}
\end{equation}
respectively.

When $\Delta q = 0$, let $J_n = \frac{M}{N}\zeta_{n,m} - \frac{1}{N}$ and follow the similar steps as above, leading to the same result (\ref{Eq: bound of bar sigma}) for $\bar{\sigma}_m^{0}$.

\section{Proof of \eqref{Eq: block circulant}}
\label{app:block circulant}

Since \eqref{Eq: block circulant} is equivalent to \eqref{Eq: block circulant2}, we prove the latter following the same technique in Appendix \ref{app:lemma circulant}.

For $q_1,q_2 \in \mathbb{N}$, $p_1,p_2 \in \mathbb{M}$, combining the definition of $\left[\bm X_{q_1,q_2}\right]_{p_1,p_2}$ in \eqref{eq:entryXq1q2} and the fact that
\begin{equation}
    \begin{split}
        e^{j2\pi\frac{q_2-q_1}{N} n} &= e^{j2\pi\frac{(q_2-q_1)-0}{N} n} \\
        &= e^{j2\pi\frac{(q_2-q_1+N)-0}{N} n},
    \end{split}
\end{equation}
yields that
$\left[\bm{X}_{q_1,q_2}\right]_{p_1,p_2} = \left[\bm{X}_{0,q_2-q_1}\right]_{p_1,p_2}$ when $q_2-q_1 \in \mathbb{N}$, and $\left[\bm{X}_{q_1,q_2}\right]_{p_1,p_2} = \left[\bm{X}_{0,q_2-q_1+N}\right]_{p_1,p_2}$ when $q_2 - q_1 + N \in \mathbb{N}$. This completes the proof.

\section{Proof of Corollary \ref{cor:psi_s}}
\label{sec:proof of spectral norm}
To calculate the singular values of $\bm \Psi$, according to \eqref{eq:sqrt eigenvalue}, we start by continuing the derivation of $\lambda_l$, $l = qM + m$, in \eqref{Eq:eigenvalue expression of X} as
\begin{equation}
\label{eq:lambdal0}
\lambda_l = \sum_{n=0}^{N-1}e^{j2\pi\frac{qn}{N}}\lambda_m^{n}
 = \frac{M}{N}\sum_{n_1=0}^{N-1}\sum_{n_2=0}^{N-1}\zeta_{n_2,m}e^{j2\pi\frac{q+n_2}{N}n_1}.
\end{equation}
With the substitute of $\zeta_{n_2,m} = \delta\left(C_{n_2} - m\right)$, \eqref{eq:lambdal0} becomes
\begin{equation}
\label{Eq: lambda_l}
\begin{split}
\lambda_l & = M\sum_{n_2=0}^{N-1}\delta \left(C_{n_2} - m\right)\left(\frac{1}{N}\sum_{n_1=0}^{N-1}e^{j2\pi\frac{n_2 + q}{N}n_1}\right)\\
&= M\sum_{n_2=0}^{N-1} \delta \left(C_{n_2} - m\right)  \delta \left((q + n_2) \mod N\right)\\
&= M\delta\left(C_{(N - q) \mod N} - m\right).
 \end{split}
\end{equation}
From \eqref{Eq: lambda_l} we find that $\lambda_l \in \{0,M\}$, which implies
\begin{equation}
    \|\bm{\Psi}\|_s = \max_{l \in \mathbb{L}}{\sigma_l\left(\bm \Psi \right)} = \max_{l \in \mathbb{L}}\sqrt{\lambda_l} = \sqrt{M}.
\end{equation}


\section{Proof of Theorem \ref{theorem:unique recovery}}
\label{sec:appendices 1}

Given a positive constant $\epsilon$, we set
\begin{equation}
 M e^{\tiny{-\frac{\left(\sqrt{N}c_1 - \sqrt{M-1}\right)^2}{4\left(M-1\right)}}} = \frac{MN}{2} e^{\tiny{-\frac{\left(\sqrt{N}c_2 - \sqrt{M-1}\right)^2}{4\left(M-1\right)}}} = \frac{\epsilon}{2},
\end{equation}
yielding
\begin{equation}
\label{eq:c1}
 c_1 = \sqrt{\frac{M-1}{N}}\left(2\sqrt{\log{2M}-\log\epsilon} + 1\right),
\end{equation}
\begin{equation}
 \label{eq:c2}
 c_2 = \sqrt{\frac{M-1}{N}}\left(2\sqrt{\log{MN}-\log\epsilon} + 1\right).
\end{equation}
Note that for a fixed $M$ and sufficiently large $N$, both $c_1$  and $c_2$ approach 0. From \eqref{Eq:uI bound} and \eqref{Eq:uB bound}, we have that $\mathds{P}\left(\mu_I > c_1\right) < \epsilon/2$ and $\mathds{P}\left(\mu_B > c_2\right) < \epsilon/2$, respectively. Applying the union bound implies
\begin{equation}
    \mathds{P} \left(\mu_I > c_1 \cup \mu_B > c_2 \right)< \frac{\epsilon}{2} + \frac{\epsilon}{2} = \epsilon.
\end{equation}

Substituting \eqref{Eq:spectral norm epxpression} into \eqref{Eq: original sufficient condition} and letting $a: = \frac{1}{8} - 24\mu_B\log{MN} - \frac{3}{2}\mu_I$, $b:=\frac{17}{4}\sqrt{\log MN \left(\mu_I + 1\right)}$, we rewrite \eqref{Eq: original sufficient condition} with some arrangement as
\begin{equation}
\label{eq:bound uI or uB}
\left( \sqrt{\frac{KM}{N}} + b \right)^2 - b^2  \le a.
\end{equation}
For $a\ge 0$, this becomes
\begin{equation}
\label{eq:bound on k to prove}
\sqrt{\frac{KM}{N}} \le \sqrt{a+b^2}-b.
\end{equation}
To find a condition such that \eqref{eq:bound on k to prove} holds, we note that $\frac{17a}{36b} \le \sqrt{a+b^2}-b$ (because $a<\frac{1}{8}$ and $\frac{17}{4} < b$) and that $\frac{17a}{36b}$ increases monotonically with the increase of $a$ (or the decrease of $b$) for $a>0$. This implies that $\frac{17a}{36b}$ is a monotonically decreasing function with respect to $\mu_I$ and $\mu_B$. We are now ready to construct a sufficient condition that makes \eqref{eq:bound on k to prove} hold.

By substituting $\mu_I = c_1$ and $\mu_B =c_2$ into $\frac{17a}{36b}$, \eqref{eq:bound on k to prove} and \eqref{Eq: original sufficient condition} hold with a probability at least $1-\epsilon$ as long as $a>0$ and
\begin{equation}
    \sqrt{\frac{KM}{N}}  \le \frac{17a}{36b}
    = \frac{\frac{1}{8}  -  24c_2\log{MN} - \frac{3}{2}c_1}{9\sqrt{\log MN \left(c_1 + 1\right)}}.
\end{equation}
With the substitution of \eqref{eq:c1} and \eqref{eq:c2}, we obtain \eqref{Eq:Block CS sparsity constrain}, completing the proof.

\ifCLASSOPTIONcaptionsoff
  \newpage
\fi

\footnotesize
\balance
\bibliographystyle{IEEEtran}
\bibliography{IEEEabrv,IEEEexample_BCS}

\end{document}